\documentclass[11pt]{article}
\usepackage[top=30truemm,bottom=30truemm,left=25truemm,right=25truemm]{geometry}
\usepackage{amssymb}
\usepackage{amsmath}
\usepackage{amsthm}
\usepackage{multirow}
\usepackage{fancybox}
\usepackage{framed}
\usepackage{graphics}

\newtheorem{definition}{Definition}
\newtheorem{theorem}{Theorem}
\newtheorem{corollary}{Corollary}
\newtheorem{lemma}{Lemma}
\newtheorem{proposition}{Proposition}

\newcommand{\calM}{\mathcal{M}}
\newcommand{\calS}{\mathcal{S}}
\newcommand{\calR}{\mathcal{R}}
\newcommand{\calA}{\mathcal{A}}
\newcommand{\calB}{\mathcal{B}}
\newcommand{\bin}{\{0,1\}}
\newcommand{\abort}{\mathsf{abort}}
\newcommand{\detect}{\mathsf{detect}}

\newcommand{\suc}{\mathsf{suc}}

\newcommand{\guess}{\mathsf{guess}}
\newcommand{\out}{\mathsf{out}}
\newcommand{\F}{\mathbb{F}}
\newcommand{\Share}{\mathsf{Share}}
\newcommand{\Reconst}{\mathsf{Reconst}}

\newcommand{\ignore}[1]{}
\mathchardef\mhyphen="2D
\newcommand{\E}{\mathbb{E}}

\newcommand{\games}{\mathsf{Game}^\mathsf{sngl}}
\newcommand{\gamem}{\mathsf{Game}^\mathsf{mult}}

\newcounter{protnum}
\newcommand{\prot}[1][]{\refstepcounter{protnum}Protocol~#1\arabic{protnum}}

\begin{document}

\title{Perfectly Secure Message Transmission against Rational Adversaries\footnote{This is the full version of~\cite{FYK18} and~\cite{YK19}.}}

\author{Maiki Fujita\thanks{Saitama University, Japan.}
  \and Takeshi Koshiba\thanks{Waseda University, Japan. Email: \texttt{tkoshiba@waseda.jp}}
 \and Kenji Yasunaga\thanks{Tokyo Institute of Technology, Japan. Email: \texttt{yasunaga@c.titech.ac.jp}}}

\date{\today}

\maketitle

\begin{abstract}
  Secure Message Transmission (SMT) is a two-party cryptographic protocol
  by which the sender can securely and reliably transmit messages to the receiver using multiple channels.
  An adversary can corrupt a subset of the channels and commit eavesdropping and tampering attacks over the channels.
  In this work, we introduce a game-theoretic security model for SMT in which adversaries have some preferences for protocol execution.
  We define rational ``timid'' adversaries  who prefer to violate security requirements but do not prefer the tampering to be detected.

  First, we consider the basic setting where a single adversary attacks the protocol.
  We construct perfect SMT protocols against any rational adversary corrupting all but one of the channels.
  Since minority corruption is required in the traditional setting, 
  our results demonstrate a way of circumventing the cryptographic impossibility results by a game-theoretic approach.
  
  Next, we study the setting in which all the channels can be corrupted by multiple adversaries who do not cooperate.
  Since we cannot hope for any security if a single adversary corrupts all the channels or multiple adversaries cooperate maliciously,
  the scenario can arise from a game-theoretic model.
  We also study  the scenario in which both malicious and rational adversaries exist.
 \end{abstract}

\section{Introduction}

It is common to use the information network to send and receive messages.
The physical channels between senders and receivers are composed of communication apparatuses,
allowing adversaries to eavesdrop or tamper.
While we usually use public-key cryptosystems to protect data over communication,
their security needs computational assumptions.
It is desirable to develop methods of achieving security in the information-theoretic sense.

In the two-party cryptographic setting,
we typically assume a single communication channel between the parties.
However, current network technologies can let many channels be available.
Secure Message Transmission (SMT),  introduced by Dolev et al.~\cite{DDWY93},
is a cryptographic protocol for securely transmitting messages through multiple channels.
Even if an adversary corrupts $t$ out of $n$ channels and commits
eavesdropping and tampering over the corrupted channels,
the messages are securely and correctly transmitted to the receiver by SMT.
The requirements for SMT consist of \emph{privacy} and \emph{reliability}.
The privacy guarantees that the adversary can obtain no information about the transmitted message,
and the reliability guarantees that the receiver recovers the message sent by the sender.
If an SMT protocol satisfies both the requirements perfectly, the protocol is called a \emph{perfect} SMT.
Spini and Z\'emor~\cite{SZ16} gave the most round-efficient perfect SMT.
Dolev et al.~\cite{DDWY93} showed that any one-round perfect SMT must 
satisfy $t < n/3$ and that any perfect SMT whose round complexity is at least two
must satisfy $t < n/2$.
Garay and Ostrovsky~\cite{GO08} introduced the model of SMT \emph{by Public Discussion} (SMT-PD),
which allows transmission over an authentic and reliable  public channel in addition to the $n$ channels.
Shi et al.~\cite{SJST11} further studied SMT-PD and constructed a round-optimal perfect SMT-PD.
In the context of \emph{network coding}, similar but more general problems have been studied,
and some schemes~\cite{YSJL14} can be seen as SMT protocols.

In the standard setting of cryptography, we assume the participants are either honest or malicious.
The former will follow the protocol description honestly,
and the latter may deviate from the protocol maliciously.
In general, malicious behavior may be illegal and involve some risks,
implying that adversaries in the standard cryptographic setting behave maliciously
regardless of their risk.
However, adversaries in real life may decide their behavior by taking the risk into account.
To capture such situations, we incorporate game-theoretic \emph{rational} participants into cryptography.
Halpern and Teague~\cite{HT04} first studied the rational behavior of participants for secret sharing.
Since then, rational secret sharing has been intensively studied~\cite{ADGH06,GK06,KN08a,KN08b,AL11,FKN10,KOTY17}.
Moreover, there have been many studies using game-theoretic analysis of cryptographic primitives/protocols,
including two-party computation~\cite{ACH16,GK12}, leader election~\cite{Gra10,ADH19},
Byzantine agreement~\cite{GKTZ12}, consensus~\cite{HV16}, public-key encryption~\cite{Yas16,YY18},
delegation of computation~\cite{AM13,GHRV14,CG15,GHRV16,CG17,IY17}, and protocol design~\cite{GKMTZ13,GKTZ15}.
Among them, several works~\cite{GKTZ12,AM13,GHRV14,GHRV16,GKMTZ13} used the rationality of adversaries
to circumvent the impossibility results.

Groce et al.~\cite{GKTZ12} studied the Byzantine agreement problem in the presence of a rational adversary.
They showed that given some knowledge of the adversary's preference,
a perfectly secure Byzantine agreement is possible for $t$ corruptions among $n$ players for any $t < n$.
The security against $t \geq n/2$ corruptions is impossible in the standard adversary model.

This work shows that the impossibility results of SMT can also be circumvented by considering adversaries' rationality.
As in the Byzantine agreement, we introduce a rational adversary for SMT who has some preference for the protocol execution outcome.
More specifically, we define \emph{timid} adversaries who prefer to violate the requirements of SMT but do not prefer the tampering  to be detected.
Such preferences can be justified if adversaries fear losing their corrupted channels when the protocol detects tampering.

\subsection{Our Results}

First, as the most basic setting, we define the security of perfect SMT in the presence of a \emph{single} rational adversary.
It is a natural extension of the standard cryptographic setting to the rational one.
We show that the almost-reliable SMT-PD protocol of~\cite{SJST11} works as a \emph{perfect} SMT protocol.
An intuitive reason is that timid adversaries do not have an incentive to attack the channels for fear of detection.
Thus perfect reliability follows as well as perfect secrecy.
To construct non-interactive SMT protocols, we consider \emph{strictly} timid adversaries
who prefer being undetected to violate the security requirements.
We show that \emph{robust} secret sharing schemes, which can detect forgery of shares with high probability,
can be used as a non-interactive SMT protocol for such  adversaries.
Both protocols are perfectly secure against timid adversaries corrupting $t$ out of $n$ channels for any $t < n$,
which is impossible in the standard setting of SMT protocols.
We also present the impossibility  of constructing SMT protocols
against general timid adversaries corrupting $t \geq n/2$ channels.
The result implies that the public discussion model is necessary for the first protocol, and the strict timidness is necessary for the second protocol.
The results are summarized in Table~\ref{tb:results}.

\begin{table}[t]
  \centering
  \caption{Summary of the Results of Single-Adversary Setting}\label{tb:results}
  \medskip
   \scalebox{0.90}{
  \begin{tabular}{cccccc}\hline
    \makebox[15mm]{\textbf{Adversary}} & \textbf{PD}$^*$ & \makebox[15mm]{\textbf{Resiliency}} & \textbf{Security} & \textbf{\#\! Round} & \makebox[35mm]{\textbf{Construction}} \\ \hline
    Malicious & --- & $t < n/3$ & Perfect & $1$ & Exist (\cite{DDWY93}) \\
    Malicious & --- & $t \geq n/3$ & Perfect & $1$ & Impossible (\cite{DDWY93})\\
    Malicious & --- & $t < n/2$ & Perfect & $2$ &Exist (\cite{KS09,SZ16})\\
    Malicious & --- & $t < n/2$ & Almost Reliable  & $1$ & Exist (\cite{YSJL14}\footnotemark)\\
    Timid & --- & $t < n/2$ & Perfect & $1$ & Exist (Corollary~\ref{cor:ciss})\\
    Malicious & --- & $t \geq n/2$ & Perfect & $2$ &Impossible (\cite{DDWY93})\\
    Malicious & \checkmark & $t \geq n/2$ & Almost Reliable & $2$ & Impossible (\cite{SJST11})\\
        Timid & --- & $t \geq n/2$ & Perfect & --- & Impossible (Corollary~\ref{cor:norsmt})\\ 
    Malicious & \checkmark & $t < n$ & Almost Reliable & $3$ & Exist (\cite{FW00,GO08,SJST11,GGO14})\\
    Timid & \checkmark & $t < n$ & Perfect & $3$ & Exist (Theorem~\ref{thm:sjst_single}) \\
    Strictly Timid & --- & $t < n$ & Perfect & $1$& Exist (Theorem~\ref{thm:sttimid}) \\\hline
    & & & & & \\
    \multicolumn{5}{l}{\raisebox{2.5ex}[0pt]{$^*$PD represents the public discussion model.}}
  \end{tabular}
  }
\end{table}

\footnotetext{The paper studied more general problems of secure network coding. By considering a simple $n$-link network as in SMT and
  assuming the adversary who eavesdrops and tampers with the same links, the coding scheme of~\cite{YSJL14} gives a construction of
an almost-reliable SMT protocol for $t < n/2$.}

Next, we study the setting in which \emph{multiple} timid adversaries may corrupt \emph{all} the channels.
More specifically, we assume that at least two adversaries exclusively corrupt subsets of the channels
and prefer to violate the security requirements without being detected.
We also assume that each adversary prefers other adversaries' tampering to be detected.
This additional assumption makes rational adversaries avoid cooperating. 
If a single adversary corrupts all the channels, we cannot hope for any security of SMT.
Thus, multiple conflicting adversaries are necessary for achieving security.
We believe that the multiple-adversary setting is more realistic than the single-adversary one
since it is difficult for the adversary to confirm that no other adversary exists.
Also, protocols in the multiple setting need to equip the property to declare the tampering detection with channel identifiers.
This required property is more desirable than a detection mechanism without channel identifiers, sufficient  in the single setting.
We show that secure SMT protocols exist 
even if such rational adversaries corrupt all the channels.
The SMT-PD protocol of~\cite{SJST11} also works in this setting as perfect SMT-PD.
To construct perfect SMT protocols without public discussion,
we employ the idea of \emph{cheater-identifiable} secret sharing (CISS),
in which every player who submits a forged share in the reconstruction phase can be identified.
We construct a non-interactive SMT protocol based on the CISS of Hayashi and Koshiba~\cite{HK18}.
Technically, our construction employs pairwise independent (a.k.a. strongly universal) hash functions instead of universal hash functions in~\cite{HK18}. 
Since the security requirements of CISS are not sufficient for proving the security of SMT against timid adversaries,
we provide the security analysis of our protocol, not for general CISS-based SMT protocols.
  The limitation of CISS is that the number of forged shares should be a minority.
  Namely, the above construction only works for adversaries who corrupt less than $n/2$ channels. 
  We show that 
  a variant of our CISS-based protocol works as a perfect SMT protocol for \emph{strictly} timid adversaries,
  even if each adversary corrupts a majority of the channels.

  Finally, we consider the setting in which a malicious adversary exists as well as rational adversaries.
  Namely, there are heterogeneous adversaries, all but one behave rationally, but one acts maliciously.
  We believe this setting is preferable because the assumption that all adversaries are rational may not be realistic.
  We show that a modification of the CISS-based protocol achieves a non-interactive perfect SMT protocol against such adversaries.
  The protocol is secure as long as a malicious adversary corrupts $t^* \leq \lfloor (n-1)/3 \rfloor$ channels,
  and each rational adversary corrupts at most $\min\{\lfloor (n-1)/2 \rfloor - t^*, \lfloor (n-1)/3 \rfloor \}$ channels.

  \begin{table}[t]
  \centering
  \caption{Perfect SMT Protocols}\label{tb:perfect_results}
  \medskip
  \scalebox{0.99}{
    \begin{tabular}{ccccccc}\hline
    \multirow{2}{*}{\textbf{Adversary}} & \multirow{2}{*}{\makebox[5mm]{\textbf{PD}}} &\textbf{Total}  & \multirow{2}{*}{\textbf{\#\! Adv.}}& \textbf{Resiliency}  & \multirow{2}{*}{\textbf{\# \!Round}}  & \multirow{2}{*}{{\textbf{References}}}\\
    &  &  \textbf{Resiliency} &  & \textbf{per Adv.} &   &  \\\hline

 %    \textbf{Adversary} & \textbf{PD} &\textbf{Total Resiliency}  & \textbf{\#\! Adv.}& \textbf{Resiliency per Adv.}  & \textbf{\# \!Round}  & \textbf{References} \\ \hline
    Malicious & ---  & $< n/3$ & $1$ & $< n/3$ & $1$ & \cite{DDWY93} \\
    Malicious & ---  & $< n/2$ & $1$ & $< n/2$ & $2$ & \cite{KS09,SZ16} \\
    Timid & --- & $< n/2$ & $1$ & $< n/2$ & $1$ & Corollary~\ref{cor:ciss} \\
    Timid & \checkmark & $< n$  & $1$ &$< n$ &  $3$ & Theorem~\ref{thm:sjst_single} \\
    Strictly Timid & --- & $< n$ & $1$ & $< n$ & $1$ & Theorem~\ref{thm:sttimid}\\
    Timid\slash Malicious & --- & $n$& $\geq 2$ &$< n/6$  & $1$ & Theorem~\ref{thm:ss-robust} \\ 
    Timid & ---  & $n$ & $\geq 2$ &$< n/2$ &  $1$ & Theorem~\ref{thm:ciss} \\
    Timid & \checkmark  & $n$ &  $\geq 2$& $< n$ &  $3$ & Theorem~\ref{thm:sjst} \\
    Strictly Timid & ---  & $n$ & $\geq 2$ &$< n$ & $1$ & Theorem~\ref{thm:ciss2} \\\hline
    \end{tabular}
}    
  \end{table}

  We summarize constructions of perfect SMT protocols both for the single-adversary and the multiple-adversary settings in Table~\ref{tb:perfect_results}.
  The total resiliency is the maximum number of corrupted channels for which the protocol can achieve security.
  
  We note that the single-adversary setting can be seen as a special case of the multiple-adversary setting.
  Some protocols for multiple adversaries may work against single adversaries.
  We show that Theorem~\ref{thm:ciss} implies a non-interactive protocol for single adversaries
  that is secure against $t < n/2$ corruption (Corollary~\ref{cor:ciss}).
  This result circumvents the impossibility of constructing one-round protocols for $t \geq n/3$ in~\cite{DDWY93}
  by a game-theoretic consideration.

\subsection{Related Work}

The adversaries' behavior of avoiding detection has been used in the literature of multiparty computation.
Franklin and Yung~\cite{FY92} defined the notion of \emph{$t$-detectability},
which guarantees that no coalition of $t$ parties can either learn any information about other $n-t$ parties' inputs
or prevent the honest parties from detecting the tampering.
Aumann and Lindell~\cite{AL10} introduced the notion of security against \emph{covert adversaries},
who attempt to cheat but do not want to be caught with some prescribed probability.
The underlying idea of covert adversaries is similar to that of timid adversaries in this work.
However, there are several key differences.
First, the goal of security notions in~\cite{FY92,AL10} is to detect the adversary's tampering.
These works do not consider what happens if tampering is detected.
We provide a game-theoretic framework that guarantees perfect secrecy and reliability against adversaries trying to avoid detection.
Second, adversaries in~\cite{FY92,AL10} only try to learn private inputs of honest parties,
while timid adversaries in this work try to violate both reliability and secrecy.
In particular, since all protocols in this work achieve perfect secrecy,
timid adversaries are essentially concerned about how to violate the reliability of the protocols.
Such adversaries were not considered in~\cite{FY92,AL10}.
Also, as far as we know, security against covert adversaries can only be achieved against computationally bounded adversaries.
We construct perfect SMT protocols against computationally unbounded adversaries.

Another efficiency metric discussed in the SMT literature is the \emph{communication complexity} of the protocols~\cite{SNR04,ACdH06}.
Minimizing the communication complexities of our protocols  is an interesting future work.

\subsection{Organization}

Section~\ref{sec:smt} describe the definitions and the known results of secure message transmission.
The security of SMT against a single adversary is given in Section~\ref{sec:singlegame},
and our protocols are presented in Section~\ref{sec:singleprotocol}.
In Section~\ref{sec:impossibility}, we show an impossibility result for general timid adversaries.
We define the security of SMT  against multiple adversaries in Section~\ref{sec:multgame}
and give the constructions of SMT protocols in Section~\ref{sec:multadv}.
We study a mixed model of rational and malicious adversaries in Section~\ref{sec:mixed}.
We conclude the paper in Section~\ref{sec:conclusion}.

\section{Secure Message Transmission}\label{sec:smt}

We assume that there are $n$ channels between a sender $\calS$ and a receiver $\calR$.
SMT protocols proceed in \emph{rounds}.
In each round, either $\calS$ or $\calR$  can send messages over the channels. 
The messages are delivered before the next round starts.
The adversary $\calA$ can corrupt at most $t$ channels out of the $n$ channels; such an adversary is referred to as \emph{$t$-adversary}.
On the corrupted channels, $\calA$ can eavesdrop, block communication, or place any messages on them.
We assume that $\calA$ is \emph{rushing}. Namely, 
$\calA$ can decide the actions on the corrupted channels after observing the information sent on the corrupted channels.
We consider computationally \emph{unbounded} $\calA$. 

Let $\calM$ be the message space.
In SMT, $\calS$ tries to transmit a message in $\calM$ to $\calR$, 
and $\calR$ outputs the received message after the protocol execution.
For an SMT protocol $\Pi$, let $M_S$ denote the random variable of the message sent by $\calS$ and
$M_R$ the message output by $\calR$.
An execution of $\Pi$ can be completely characterized by the random coins of all the parties, namely, $\calS$, $\calR$, and $\calA$,
and the message $M_S$.
Let $V_A(m, r_A)$ be the \emph{view} of $\calA$ when $M_S = m$, and $\calA$ uses $r_A$ as the random coins.
Precisely, $V_A(m, r_A)$ consists of the messages sent over the corrupted channels when the protocol is run with $M_S = m$,
and $r_A$, the random coins of $\calA$.

We formally define the security requirements of SMT protocols.
\begin{definition}
  A protocol between $\calS$ and $\calR$ is \emph{$(\varepsilon, \delta)$-Secure Message Transmission (SMT)
  against $t$-adversary}
  if the following three conditions are satisfied against any $t$-adversary $\calA$:
  \begin{itemize}
  \item \emph{Correctness}: For any $m \in \calM$,
    if $M_S = m$ and $\calA$ does not change messages sent over the corrupted channels, then $\Pr[ M_R = m ] = 1$.
  \item \emph{Privacy}: For any $m_0, m_1 \in \calM$ and $r_A \in \bin^*$, it holds that
    \begin{equation*}
      \Delta( V_A(m_0, r_A), V_A(m_1, r_A) ) \leq \varepsilon,
    \end{equation*}
    where $\Delta(X, Y)$ denotes the statistical distance between two random variables $X$ and $Y$ over a finite set $\Omega$,
    which is defined by
    \begin{equation*}
      \Delta(X, Y) = \frac{1}{2} \sum_{u \in \Omega} \left| \Pr[X = u] - \Pr[Y = u] \right|.
    \end{equation*}

  \item \emph{Reliability}: For any message $m \in \calM$, when $M_S = m$,
    \begin{equation*}
      \Pr[ M_R \neq m ] \leq \delta,
    \end{equation*}
          where the probability is taken over the random coins of $\calS$, $\calR$, and $\calA$.
  \end{itemize}
\end{definition}

A protocol achieving $(0,0)$-SMT is called \emph{perfect}. 
If a protocol achieves $(0,\delta)$-SMT for small $\delta$, it is called \emph{almost-reliable} SMT.

Dolev~et~al.~\cite{DDWY93} characterized the trade-off between the achievability and the round complexity of perfect SMT.
\begin{theorem}[\cite{DDWY93}]\label{thm:ddwy}
  One-round perfect SMT protocols against $t$-adversary exist if and only if $t < n/2$.
  Also, multi-round perfect SMT protocols against $t$-adversary exist if and only if $t < n/3$.
\end{theorem}

\paragraph*{SMT by Public Discussion.}
In addition to the $n$ channels, we may assume that 
$\calS$ and $\calR$ can use an authentic and reliable \emph{public channel}
on which messages are publicly accessible and guaranteed to be correctly delivered.
Such protocols are referred to as SMT \emph{by Public Discussion} (SMT-PD).
Franklin and Wright~\cite{FW00} gave an impossibility result of SMT-PD by using different terminology.
(See~\cite{GGO11} for this fact.)
\begin{theorem}[\cite{FW00}]\label{thm:fw}
  Perfectly-reliable ($\delta=0$) SMT-PD protocols against $t$-adversary exist only if $t < n/2$.
\end{theorem}

Shi~et~al.~\cite{SJST11} gave several impossibility results of SMT-PD and constructed a round-optimal SMT-PD protocol.
We use their protocol. The description appears in Section~\ref{sec:single_sjst}.

\section{SMT against a Single Rational Adversary}\label{sec:singlegame}

We define our security model of SMT protocol against a single rational adversary.
The rationality of the adversary is characterized by a \emph{utility function}
that represents the preference of the adversary over possible outcomes of the protocol execution.

We can consider various preferences of the adversary regarding the SMT protocol execution.
The adversary may prefer to violate the privacy or the reliability of SMT protocols.
Also, the adversary may prefer to violate the above properties without being tampering detected. 
Here, we consider the adversary who prefers (1) to violate privacy, (2) to violate reliability, and (3) the tampering to be undetected.

To define the utility function, we specify the SMT game as follows.
\paragraph*{The SMT Game.}
For an SMT protocol $\Pi$, we define our SMT game $\games(\Pi, \calA)$ against a single adversary $\calA$.
First, set three parameters $\guess = \suc = \detect  = 0$.
For the message space $\calM$, choose $m \in \calM$ uniformly at random,
and run the protocol $\Pi$ in which the message to be sent is $M_S = m$.
In the protocol execution, as in the usual SMT, the adversary $\calA$ can corrupt at most $t$ channels
and tamper with any messages sent over the corrupted channels.
If the sender or the receiver sends a special message ``DETECT'' during the execution, set $\detect = 1$.
After running the protocol, the receiver outputs $M_R$, and the adversary outputs $M_A$.
If $M_R = M_S$, set $\suc = 1$.
If $M_A = M_S$, set $\guess = 1$.
The outcome of the game is $(\guess, \suc, \detect)$.

By following~\cite{FKN10}, we model the game where the adversary tries to guess the message
chosen uniformly at random.
In general, it is difficult to model the ``real'' game that the adversary attacks.
The above formulation can capture the situation that the adversary learns partial information of the message.
If the partial information increases the probability of correctly guessing the message, the adversary obtains higher utility in the game.

The utility of the adversary is defined as the expected utility in the SMT game.

\begin{definition}[Utility]
  The utility $u(\calA,U)$ of the adversary $\calA$ with utility function $U$
  is the expected value $\E[U(\out)]$, where $U$ is a function that maps the outcome $\out = (\guess, \suc, \detect)$ of the game $\games(\Pi, \calA)$ 
  to real values and the probability is taken over the randomness of the game.
\end{definition}

The utility function $U$ characterizes the type of adversaries.
If the adversary has the preferences (1)-(3) as above, the utility function may have the property  such that
for any two outcomes $\out = (\guess, \suc, \detect)$ and $\out' = (\guess', \suc', \detect')$ of the SMT game,
\begin{enumerate}
\item $U(\out) > U(\out')$ if $\guess > \guess'$, $\suc = \suc'$, and $\detect = \detect'$;
\item $U(\out) > U(\out')$ if $\guess = \guess'$, $\suc < \suc'$, and $\detect = \detect'$;
\item $U(\out) > U(\out')$ if $\guess = \guess'$, $\suc = \suc'$, and $\detect < \detect'$.
\end{enumerate}

Based on the utility function of the adversary, we define the security of SMT against rational adversaries.
In particular, regarding the security requirements, we only consider perfect SMT.
\begin{definition}[PSMT against a Rational Adversary]\label{def:rsmtsec}
  An SMT protocol $\Pi$ is \emph{perfectly secure against a rational $t$-adversary with utility function $U$} if
  there is a $t$-adversary $\calB$  such that 
  \begin{enumerate}
  \item \emph{Perfect security}: $\Pi$ is $(0,0)$-SMT against $\calB$; and
  \item \emph{Nash equilibrium}: $u(\calA,U) \leq u(\calB,U)$ for any $t$-adversary $\calA$. 
  \end{enumerate}
\end{definition}

The perfect security guarantees that an adversary $\calB$ is \emph{harmless}.
The Nash equilibrium guarantees that no adversary $\calA$ can gain more utility than $\calB$.
Thus, the above security implies that no adversary $\calA$ can gain more utility
than the harmless adversary.
Namely, the adversary does not have an incentive to deviate from the strategy of the harmless adversary $\calB$.

In the security proof of our protocol, we will consider an adversary $\calB$ 
who does not tamper with any messages on the channels and outputs a random message from $\calM$ as $M_A$.
We call such $\calB$ a \emph{random guessing} adversary.
If $\calB$ is random guessing, then the perfect security immediately follows from the correctness property of $\Pi$.

\subsection*{Timid Adversaries}

We construct secure protocols against \emph{timid} adversaries,
who prefer to violate the security requirements of SMT protocols and do not prefer the tampering  to be detected.
More formally, the utility function  of such adversaries  should have properties such that
\begin{enumerate}
\item $U(\out) > U(\out')$ if  $\guess = \guess'$, $\suc < \suc'$, and $\detect = \detect'$; and
\item $U(\out) > U(\out')$ if  $\guess = \guess'$, $\suc = \suc'$, and $\detect < \detect'$,
\end{enumerate}
where $\out = (\guess, \suc, \detect, \abort)$ and $\out' = (\guess', \suc', \detect', \abort')$ are the outcomes of the SMT game.
Let $U_\mathsf{timid}$ be the set of utility functions that satisfy the above conditions.

Also, timid adversaries may prefer being undetected to violating security.
Such adversaries have the following utility:
\begin{enumerate}
 \setcounter{enumi}{2}
\item $U(\out) > U(\out')$ if  $\guess = \guess'$, $\suc > \suc'$, and $\detect < \detect'$.
\end{enumerate}
Let $U_\mathsf{st\mhyphen timid}$ be the set of utility functions satisfying the above three conditions.
An adversary is  \emph{timid} if his utility function is in $U_\mathsf{timid}$,
and \emph{strictly timid} if the utility function is in $U_\mathsf{st\mhyphen timid}$.

In the analysis of our protocols, we need the following four values of utility:
\begin{itemize}
\item $u_1$ is the utility when $\Pr[\guess = 1] = \frac{1}{|\calM|}$,   $\suc = 0$, and $\detect = 0$;
\item $u_2$ is the utility when $\Pr[\guess = 1] = \frac{1}{|\calM|}$,   $\suc = 1$, and $\detect = 0$;
\item $u_3$ is the utility when $\Pr[\guess = 1] = \frac{1}{|\calM|}$,   $\suc = 0$, and $\detect = 1$;
\item $u_4$ is the utility when $\Pr[\guess = 1] = \frac{1}{|\calM|}$,   $\suc = 1$, and $\detect = 1$.
\end{itemize}
It follows from the properties of utility functions in $U_\mathsf{timid}$ that 
$u_1 > \max \{u_2, u_3\}$ and $\min\{u_2, u_3\} > u_4$.
For utility functions in $U_\mathsf{st\mhyphen timid}$, it holds that $u_1 > u_2 > u_3 > u_4$.

The fact that all the utilities we use in our analysis are $u_1, u_2, u_3, u_4$ implies that
all of  our protocols achieve \emph{perfect} privacy
since the probability of guessing the message is $1/|\calM|$ in every utility.

\section{Protocols against a Timid Adversary}\label{sec:singleprotocol}

\subsection{Protocol by Public Discussion}\label{sec:single_sjst}

We show that an almost-reliable SMT-PD protocol proposed by Shi, Jiang, Safavi-Naini, and Tuhin~\cite{SJST11}
works as a perfect SMT-PD protocol against a timid adversary.

First, we describe the protocol and its proof overview.
\paragraph{The SJST Protocol.}
The protocol is based on the simple protocol for \emph{static} adversaries
where the sender sends a random key $R_i$ over the $i$th channel for each $i \in \{1,\dots,n\}$,
and the encrypted message $c = m \oplus R_1 \oplus \dots \oplus R_n$ over the public channel.
Suppose that the adversary $\calA$ sees the messages sent over the corrupted channels and does not change them.
Since $\calA$ cannot see at least one key $R_j$ when corrupting less than $n$ channels,
the mask $R_1 \oplus \dots \oplus R_n$ for the encryption looks random for $\calA$.
Thus, the message $m$ can be securely encrypted and reliably sent through the public channel.
The SJST protocol employs a mechanism for detecting the adversary's tampering by using hash functions
to cope with \emph{active} adversaries, who may change messages sent over the corrupted channels.
Specifically, the \emph{pairwise independent} hash functions (see Appendix~\ref{sec:uh}) satisfy the following property:
when a pair of keys $(r_i, R_i)$ is changed to $(r_i',R_i') \neq (r_i, R_i)$,
the hash value for $(r_i,R_i)$ is different from that for $(r_i',R_i')$ with high probability if the hash function is chosen randomly
after the tampering occurred.
In the SJST protocol, the sender sends a pair of keys $(r_i, R_i)$ over the $i$th channel.
Then, the receiver chooses $n$ pairwise independent hash functions $h_i$'s, and sends them over the public channel.
By comparing hash values for $(r_i,R_i)$'s sent by the sender with those for $(r_i',R_i')$'s received by the receiver,
they can identify the channels for which messages, i.e., keys, were tampered.
By ignoring keys sent over such channels, the sender can correctly encrypt a message $m$ with untampered keys 
and send the encryption reliably over the public channel.

We give a formal description of the SJST protocol in Figure~\ref{fig:sjst},
a three-round protocol that achieves
reliability with $\delta = (n-1)\cdot 2^{1-\ell}$,
where $\ell$ is the length of hash values.

\begin{figure}[t]
  \begin{framed}
    Let $n$ be the number of channels, $m \in \calM$ the message to be sent by the sender $\calS$, and
  $H = \{h \colon \{0,1\}^k \rightarrow \{0, 1\}^\ell\}$ a class of pairwise independent hash functions.
  \begin{enumerate}
  \item For each $i \in \{1, \dots, n\}$, $\calS$ chooses $r_i \in \bin^\ell$ and $R_i \in \bin^k$ uniformly at random,
    and sends the pair $(r_i, R_i)$ over the $i$th channel.
  \item For each $i \in \{1, \dots, n\}$, $\calR$ receives $(r_i',R_i')$ through the $i$th channel,
    and then chooses $h_i \leftarrow H$ uniformly at random.
    If  $|r_i'| \neq \ell$ or $|R_i'| \neq k$, set $b_i = 1$, and otherwise, set $b_i = 0$.
    Then, set $T_i' = r_i' \oplus h_i(R_i')$, and $H_i = (h_i, T_i')$ if $b_i=0$, and $H_i = \bot$ otherwise.
    Finally, $\calR$ sends $(B, H_1, \dots, H_n)$ over the public channel, where $B = (b_1, \dots, b_n)$.
  \item $\calS$ receives $(B, H_1, \dots, H_n)$ through the public channel.
    For each $i \in \{1, \dots, n\}$ with $b_i = 0$, $\calS$ computes $T_i = r_i \oplus h_i(R_i)$,
    and sets $v_i = 0$ if $T_i = T_i'$, and $v_i = 1$ otherwise.
    Then, $\calS$ sends $(V, c)$ over the public channel, where $V = (v_1, \dots, v_n)$,
    and $c = m \oplus (\bigoplus_{v_i=0}R_i)$.
  \item On receiving $(V,c)$, $\calR$ 
    recovers $m = c \oplus (\bigoplus_{v_i=0}R_i)$.
  \end{enumerate}
  \end{framed}
  \caption{The SJST Protocol}\label{fig:sjst}
\end{figure}

\begin{theorem}[\cite{SJST11}]
The SJST protocol is $(0,(n-1)\cdot 2^{1-\ell})$-SMT against $t$-adversary for any $t < n$.
\end{theorem}

One can find a complete proof of the above theorem in \cite{SJST11}.
For self-containment, we  give a brief sketch of the proof.
\begin{itemize}
\item \emph{Privacy}:
The adversary can get $c = m \oplus (\bigoplus_{v_i=0} R_i)$ through the public channel.
Since $m$ is masked by uniformly random $R_i$'s, the adversary has
to corrupt all the $i$th channels with $v_i=0$ to recover $m$.
However, since any $t$-adversary can corrupt at most $t$ $(< n)$ channels,
the adversary can cause $v_i = 1$ for at most $n-1$ $i$'s.
There is at least one $i$ with $v_i = 0$, for which
the adversary cannot obtain $R_i$.
Thus, the protocol satisfies the perfect privacy.
\item \emph{Reliability}:
  Since the protocol uses  the public channel in the second and the third rounds,
  the adversary can tamper with channels only in the first round.
  Suppose that the adversary tampers with $(r_i,R_i)$.
  If $R_i \ne R_i'$ and $T_i = T_i'$, then
  $\calR$ would recover a wrong message, but the tampering is not detected.
  The property of pairwise independent hash functions (Appendix~\ref{sec:uh}) implies that 
  the above event happens with probability at most $(n-1)2^{1-\ell}$.
  Thus, the protocol achieves reliability with $\delta = (n-1)\cdot 2^{1-\ell}$. 
\end{itemize}

For our purpose, we slightly modify the SJST protocol such that in the second and the third rounds,
if $b_i = 1$ in $B$ or $v_j = 1$ in $V$ for some $i, j \in \{1, \dots, n\}$,
the special message ``DETECT'' is also sent.
We clarify the parameters of the SJST protocol to work as SMT against timid adversaries.

\begin{theorem}\label{thm:sjst_single}
  If the parameter $\ell$ in the SJST protocol satisfies
  \begin{equation*}
    \ell \geq \max\left\{1 + \log_2 t + \log_2 \frac{u_3 - u_4}{u_2 - u_4 - \alpha}, 1 + \frac{1}{t} \log_2 \frac{u_1-u_3}{\alpha} \right\}
  \end{equation*}
  for some  $\alpha \in (0, u_2 - u_4)$,
  then the protocol is perfectly secure against a rational $t$-adversary with utility function $U \in U_\mathsf{timid}$ for any $t < n$.
\end{theorem}
\begin{proof}
  The perfect security of Definition~\ref{def:rsmtsec} immediately follows by letting $\calB$ be a random guessing adversary.
  We show that the strategy of $\calB$ is a Nash equilibrium.
  Note that $u(\calB, U) = u_2$, since $\Pr[ \guess = 1] = \Pr[ M_A = M_S] = 1/|\calM|$ in the SMT game.
  Thus, it is sufficient to show that $u(\calA, U) \leq u_2$ for any $t$-adversary $\calA$.
  Also, note that, since the SJST protocol achieves the perfect privacy,
  it holds that $\Pr[ \guess = 1] = 1/|\calM|$ for any $t$-adversary.
  
  Messages in the second and the third rounds are sent through the public channel.
  Thus, $\calA$ can tamper with messages only in the first round.
  If $\calA$ changes the lengths of $r_i$ and $R_i$, the tampering of the $i$th channel will be detected.
  Such channels are simply ignored in the second and third rounds. Thus, such tampering cannot increase the utility.
  Hence, we assume that $\calA$ does not change the lengths of $r_i$ and $R_i$ in the first round.

  Suppose that $\calA$ corrupts some $t$ channels in the first round.
  Namely, there are exactly $t$ distinct $i$'s such that $(r_i', R_i') \neq (r_i, R_i)$.
  %If $r_i' \neq r_i$ and $R_i' = R_i$,
  Note that the tampering on the $i$th channel such that $r_i' \neq r_i$ and $R_i' = R_i$ does not increase the probability that $\suc = 0$,
  but may increase the probability of detection.
  Thus, we also assume that $R_i' \neq R_i$ for all the corrupted channels.
  We define the following three events:
  \begin{itemize}
  \item $E_1$: No tampering is detected in the protocol;
  \item $E_2$: At least one but not all tampering actions are detected;
  \item $E_3$: All the $t$ tampering actions are detected.
  \end{itemize}
  Note that all the events are disjoint, and either event should occur. Namely, we have that $\Pr[E_1] + \Pr[E_2] + \Pr[E_3] = 1$.
  It follows from the discussion in Appendix~\ref{sec:uh} that the probability that
  the tampering action on one channel is not detected is $2^{1-\ell}$.
  Since each hash function $h_i$ is chosen independently for each channel,  we have that $\Pr[ E_1] = 2^{(1-\ell)t}$.
  Similarly, we obtain that $\Pr[ E_3 ] = (1 - 2^{1-\ell})^t$.
  Note that the utility when $E_1$ occurs is at most $u_1$.
  Also, the utilities when $E_2$ and $E_3$ occur are at most $u_3$ and $u_4$, respectively.
  Therefore, the utility of $\calA$ satisfies
  \begin{align}
    u(\calA, U) & \leq u_1 \cdot \Pr[E_1] + u_3 \cdot \Pr[E_2]  + u_4 \cdot \Pr[E_3] \nonumber\\
    & = u_3 + (u_1 - u_3) \Pr[E_1]- (u_3 - u_4)\Pr[E_3]\nonumber\\
    & \leq u_3 + (u_1 - u_3) 2^{(1-\ell)t} - (u_3 - u_4)\left(1- t 2^{1-\ell}\right)\nonumber\\
    & \leq u_3 + \alpha - (u_3 - u_4)\left(1- t 2^{1-\ell}\right)\label{eq:1}\\
    & \leq u_2,\label{eq:2}
  \end{align}
  where we use the relations $\ell \geq 1 + \frac{1}{t} \log_2 \frac{u_1-u_3}{\alpha}$ and $\ell \geq 1 + \log_2 t + \log_2 \frac{u_3 - u_4}{u_2 - u_4 - \alpha}$
  in (\ref{eq:1}) and (\ref{eq:2}), respectively.
  The utility of $\calA$ is at most $u_2$, and hence the statement follows.
\end{proof}

%For strictly timid adversaries, namely,
If $u_2 > u_3$, which holds for strictly timid adversaries,
by choosing $\alpha = u_2 - u_3$, the condition on $\ell$ is that
\begin{equation*}
  \ell \geq \max\left\{ 1 + \log_2 t, 1 + \frac{1}{t} \log_2 \frac{u_1 - u_3}{u_2-u_3} \right\}.
  \end{equation*}

\subsection{Protocol against a Strictly Timid Adversary}\label{sec:strict}

We show that, under the condition that $u_2 > u_3$,
a \emph{robust secret sharing scheme} 
gives a non-interactive perfect SMT protocol.
Namely, we can construct a non-interactive protocol for strictly timid adversaries.

\subsubsection{Robust Secret Sharing}\label{sec:rss}

\emph{Secret sharing}, introduced by Shamir~\cite{Sha79} and Blackley~\cite{Bla79},
enables us to distribute secret information securely.
Let $s \in \F$ be a secret from some finite field $\F$.
A (threshold) secret-sharing scheme provides a way for distributing $s$ into $n$ shares $s_1, \dots, s_n$ such that,
for some parameter $t > 0$,
(1) any $t$ shares give no information about $s$, and
(2) any $t+1$ shares uniquely determine $s$.

\begin{definition}
  Let $t,n$ be positive integers with $t < n$.
  A \emph{$(t, n)$-secret sharing scheme} with range $\mathcal{G}$ consists of
  two algorithms $(\Share, \Reconst)$ satisfying the following conditions:
  \begin{itemize}
  \item \emph{Correctness}: For any $s \in \mathcal{G}$ and $I \subseteq \{1, \dots, n\}$ with $|I| > t$,
    \begin{equation*}
      \Pr \left[ (\tilde{s},J) \gets \Reconst\left( \{i, s_i\}_{i \in I}\right) \wedge \tilde{s} = s \right] = 1, 
    \end{equation*}
    where $(s_1, \dots, s_n) \leftarrow \Share(s)$, and
  \item \emph{Perfect privacy}: For any $s, s' \in \mathcal{G}$ and $I \subseteq \{1, \dots, n\}$ with $|I| \leq t$,
    \begin{equation*}
      \Delta\left( \{s_i\}_{i \in I}, \{s'_i\}_{i \in I}\right) = 0,
    \end{equation*}
    where $(s_1, \dots, s_n) \leftarrow \Share(s)$ and $(s_1', \dots, s_n') \leftarrow \Share(s')$.
    \end{itemize}
  \end{definition}

Shamir~\cite{Sha79} gave a $(t,n)$-secret sharing scheme based on polynomial evaluations for any $t < n$.
Let $\F$ be a finite field of size at least $n$.
Then, for a given secret $s \in \F$, the sharing algorithm chooses random elements $r_1, \dots ,r_t \in \F$,
and constructs a polynomial $f(x) = s + r_1x+r_2x^2+ \dots + r_t x^t$ of degree $t$ over $\F$.
Then, for a fixed set of $n$ distinct elements $\{a_1, \dots, a_n\} \subseteq \F$,
the $i$th share is $f(a_i)$ for $i \in \{1,\dots,n\}$.
Given $\{i, f(a_i)\}_{i \in I}$ for $|I| > t$, the reconstruction algorithm
recovers 
$f$ by polynomial interpolation, and outputs $f(0) = s$ as a recovered secret.

McEliece and Sarwate~\cite{MS81} observed that Shamir's scheme is closely related to Reed-Solomon codes,
and thus the shares can be efficiently recovered even if some of them have been tampered with.
We use the fact that even if at most $\lfloor (n-1)/3 \rfloor$ out of the $n$ shares are tampered with,
the original secret can be correctly recovered by decoding algorithms of Reed-Solomon codes.
This property is called \emph{robustness}.
Although robustness is a desirable property, it is known that robust secret sharing is impossible when $t/2$ shares are tampered with~\cite{IOS12}.

In this work, we need a weaker notion of robustness in which any tampering actions should be detected with high probability.
Such robust secret sharing was studied by Cramer et al.~\cite{CDFPW08}.
They introduced the notion of \emph{algebraic manipulation detection (AMD) codes}
and presented a simple way for constructing robust secret sharing from \emph{linear} secret sharing and AMD codes.
The robustness required for our protocol is slightly different from the one defined in~\cite{CDFPW08}\footnote{The robustness
  in~\cite{CDFPW08} requires that the output of the reconstruction algorithm should be either the original message or the failure symbol with high probability.
  Namely, it is allowed to recover the original message even if some shares are tampered with.
In Definition~\ref{def:robustss}, we require that if some shares are tampered with, the output of the reconstruction algorithm should be the failure symbol.}.

\begin{definition}\label{def:robustss}
  Let $t, n$ be positive integers with $t < n$.
  A \emph{$(t, n, \delta)$-robust secret sharing} scheme with range $\mathcal{G}$
  consists of two algorithms $(\Share, \Reconst)$ satisfying the following conditions:
  \begin{itemize}
  \item \emph{Correctness}: For any $s \in \mathcal{G}$ and $I \subseteq \{1, \dots, n\}$ with $|I| > t$,
    \begin{equation*}
      \Pr \left[ \Reconst\left( \{i, s_i\}_{i \in I}\right) = s \right] = 1, 
    \end{equation*}
    where $(s_1, \dots, s_n) \leftarrow \Share(s)$.
  \item \emph{Perfect Privacy}: For any $s, s' \in \mathcal{G}$ and $I \subseteq \{1, \dots, n\}$ with $|I| \leq t$,
    \begin{equation*}
      \Delta\left( \{s_i\}_{i \in I}, \{s'_i\}_{i \in I}\right) = 0,
    \end{equation*}
    where $(s_1, \dots, s_n) \leftarrow \Share(s)$ and $(s_1', \dots, s_n') \leftarrow \Share(s')$.
  \item \emph{Robustness}: For any $s \in \mathcal{G}$ and $I \subseteq \{1, \dots, n\}$ with $|I| \leq t$ and adversary $\calA$,
    if $\tilde{s}_i \neq s_i$ for some $i \in \{1, \dots, n\}$,
    \begin{equation*}
      \Pr \left[ \Reconst\left( \{i, \tilde{s}_i\}_{i \in \{1,\dots,n\}}\right) \neq\bot \right] \leq \delta, 
    \end{equation*}
    where
    \begin{equation*}
      \tilde{s}_i =
      \begin{cases}
        \calA(i, s, \{s_i\}_{i \in I}) & \text{if } i \in I\\
        s_i & \text{if } i \notin I
      \end{cases}  
    \end{equation*}
    and $(s_1, \dots, s_n) \leftarrow \Share(s)$.
  \end{itemize}
\end{definition}

We can see that the construction of~\cite{CDFPW08} satisfies the above definition.
Specifically, we have the following theorem, which will be used in our protocol in Section~\ref{sec:robustsubsec}.
See Appendix~\ref{sec:proofrobustss} for the proof.
\begin{theorem}\label{thm:robustss}
  Let $\F$ be a finite field of size $q$ and characteristic $p$, and $d$ an integer such that $d+2$ is not divisible by $p$.
  For any positive integers $t$ and $n$ satisfying $t < n \leq qd$, there is an explicit and efficient scheme of $(t, n, (d+1)/q)$-robust secret sharing
  with range $\F^d$, where each share is an element of $\F^{d+2}$.
\end{theorem}

\subsubsection{Our Protocol}\label{sec:robustsubsec}

Let $(\Share, \Reconst)$ be a $(t,n,\delta)$-robust secret sharing scheme with range $\calM$.
In the protocol, given a message $m \in \calM$, the sender generates $n$ shares $(s_1, \dots, s_n)$ by $\Share(m)$,
and sends each $s_i$ over the $i$th channel.
The receiver simply recovers the message by $\Reconst(\{i, \tilde{s}_i\}_{i \in \{1,\dots,n\}})$, where $\tilde{s}_i$ is the received message over the $i$th channel.

\begin{theorem}\label{thm:sttimid}
  The above protocol using a $(t, n, \delta)$-robust secret sharing scheme
  is perfectly secure against a rational $t$-adversary with utility function $U \in U_\mathsf{st\mhyphen timid}$
  if $U$ satisfies  $u_2 > u_3$ and 
  \begin{equation*}
    \delta \leq \frac{u_2 - u_3}{u_1 - u_3}.
  \end{equation*}
\end{theorem}
\begin{proof}
  As in the proof of Theorem~\ref{thm:sjst_single}, we consider a random guessing adversary $\calB$.
  Then, the perfect security immediately follows.

  We show that for any $t$-adversary $\calA$, $u(\calA,U) \leq u(\calB,U)$.
  As discussed in the proof of Theorem~\ref{thm:sjst_single}, it is sufficient to prove that $u(\calA, U) \leq u_2$ for any $\calA$.
  Since the underlying secret sharing has the perfect privacy, we have that $\Pr[ \guess = 1] = 1/|\calM|$ for any $t$-adversary.
  Suppose $\calA$ corrupts some $t$ channels and alters some messages $s_i$ into different $\tilde{s}_i$.
  It follows from the robustness of secret sharing that the tampering  is detected with probability at least $1-\delta$,
  in which case the secret is not recovered.
  Thus, the utility of $\calA$ is
  \begin{align}
    u(\calA, U) & \leq (1 - \delta) u_3 + \delta u_1  \leq u_2, \label{eq:3}
  \end{align}
  where (\ref{eq:3}) follows from the assumption.
  Therefore, the statement follows.
\end{proof}

The following corollary immediately follows.
\begin{corollary}
  Let $\F$ be a finite field of size $q = 2^\ell$, and $d$ be any odd integer.
  The non-interactive protocol based on Theorem~\ref{thm:robustss}
  is an SMT protocol with message space $\F^d$ that 
  is perfectly secure against a rational $t$-adversary with utility function
  $U \in U_\mathsf{st\mhyphen timid}$ for any $t < n \leq 2^\ell d$ if
  \[   \ell \geq \log_2 (d+1) + \log_2\frac{u_1 - u_3}{u_2 - u_3}.\]
\end{corollary}

\section{Impossibility Result for General Timid Adversaries}\label{sec:impossibility}

We show that no SMT protocol is secure against a general timid $t$-adversary for $t \geq n/2$
without the public channel.
The result implies that using the public channel in Theorem~\ref{thm:sjst_single} is necessary for achieving $t \geq n/2$.
It also demonstrates the necessity of restricting the utility in Theorem~\ref{thm:sttimid}
for constructing protocols for $t \geq n/2$ without using the public channel.

\begin{theorem}\label{thm:impossibility}
  For any SMT protocol without using the public channel
  that is perfectly secure against a rational $t$-adversary with utility function $U \in U_\mathsf{timid}$,
  if $U$ has the relation
  \[ u_2 < \frac{1}{2}\left( 1 - \frac{1}{|\calM|}\right) u_3 \]
  then $t < n/2$, where $\calM$ is the message space of the protocol.
\end{theorem}
\begin{proof}
  Let $\Pi$ be a protocol in the statement.
  We construct a $t$-adversary $\calA$ for $t = \lceil n/2 \rceil$ that can successfully attack $\Pi$.
  For simplicity, we assume that $n = 2t$.
  
  Let $\calB$ be a random guessing adversary. 
  Since $\Pi$ is $(0,0)$-SMT against $\calB$, it holds that $u(\calB,U) \leq u_2$.
  We show the existence of a $t$-adversary $\calA$ that achieves $u(\calA,U) > u_2$,
  which implies that $\Pi$ cannot achieve a Nash equilibrium.

  In the SMT game, a message $m \in \calM$  is randomly chosen, and, on input $m$,
  $\Pi$ generates $(s_1^j, \dots, s_n^j)$ for $j = 1, \dots$,
  where $s_i^j$ is the message to be sent over the $i$th channel in the $j$th round.
  In the game, $\calA$ does the following:
  \begin{itemize}
  \item Randomly choose $I \subseteq \{1, \dots, n\}$ such that $|I| = t$, and corrupt the $i$th channel for every $i \in I$.
  \item Randomly choose $\tilde{m} \in \calM$, and simulate $\Pi$ on input $\tilde{m}$.
     Let $\tilde{s}_i^j$ be the message generated for the $i$th channel in the $j$th round.
  \item In each round $j$, for every $i \in I$, on receiving $s_i^j$ through the $i$th channel,
    exchange $s_i^j$ for $\tilde{s}_i^j$.
  \end{itemize}
  For this attack,
  the receiver cannot distinguish which message, $m$ or $\tilde{m}$,
  was originally transmitted by the sender
  since both messages for $m$ and $\tilde{m}$ are equally mixed.
  Hence, the probability that $\suc = 1$, denoted by $p_s$, is at most 
  \[ p_s  \leq \frac{1}{2}\left(1 - \frac{1}{|\calM|} \right) + \frac{1}{|\calM|}
  = \frac{1}{2}\left(1 + \frac{1}{|\calM|} \right), \]
  where $1/|\calM|$ comes from the event that $\tilde{m} = m$.

  Let $p_d$ be the probability that $\Pi$ outputs ``DETECT'' messages during the execution
  against the above attack.
  Without loss of generality, we assume that if $\Pi$ does not output ``DETECT'' messages,
  the receiver outputs some message at the end of the protocol.
  If the tampering  of $\calA$ is not detected,
  the utility of $\calA$ is at least $u_1$ with probability $1 - p_s$, and at least $u_2$ with probability $p_s$.
  If some tampering is detected, 
  there can be two cases:  (1) the receiver does not output any message, and (2) the receiver outputs some message.
  In case (1), the utility of $\calA$ is $u_3$.
  In case (2), the probability that the $\suc = 1$ is at most $p_s$ by the same argument as above.
  Hence, the utility of $\calA$ when the tampering was detected is at least $(1-p_s)u_3$.
  Thus, the utility of $\calA$ in the SMT game is at least
  \begin{align}
    u(\calA,U) & \geq (1 - p_d) \left( (1-p_s) u_1 + p_s u_2 \right) + p_d  (1-p_s) u_3 \nonumber \\
    & = (1-p_s) u_1 + p_s u_2 - p_d \left( (1 - p_s)u_1 + p_s u_2 - (1-p_s)u_3 \right) \nonumber \\
    & \geq (1 - p_s) u_3 \label{eq:4} \\
    & \geq \frac{1}{2}\left( 1 - \frac{1}{|\calM|}\right) u_3 \nonumber\\
    & > u_2, \label{eq:5}
  \end{align}
  where (\ref{eq:4}) follows from the fact that $p_d \leq 1$ and $(1-p_s)u_1 + p_s u_2 - (1-p_s)u_3 \geq 0$,
  and the assumption on $U$ is used in (\ref{eq:5}).
  Therefore, $\Pi$ does not satisfy the PSMT security for $t \geq n/2$. 

  When $n = 2t-1$, the same attack as the above $\calA$ can be implemented
  by   invalidating the $n$th channel 
  by substituting $\bot$ for every message over the $n$th channel.  
\end{proof}

The theorem gives the following corollary.
\begin{corollary}\label{cor:norsmt}
  There is no SMT protocol without a public channel that is perfectly secure against
  a rational $t$-adversary with utility function $U$ for every $U \in U_\mathsf{timid}$
  and $t \geq \lceil n/2 \rceil$.
\end{corollary}

\section{SMT against Multiple Rational Adversaries}\label{sec:multgame}

We define our security model of SMT in the presence of multiple rational adversaries.
For simplicity, we assume that each adversary corrupts different channels.
A difference from the single-adversary model in Section~\ref{sec:singlegame} is that
each adversary may prefer the tampering of other adversaries to be detected.
The SMT game is slightly changed such that the protocol needs to declare the tampering detection
with channel identifiers.
With this functionality, adversaries will notice the detection of their tampering.

 Suppose that there are $\lambda$ adversaries $1, 2, \dots, \lambda$ for $\lambda \geq 2$
 and 
adversary $j \in \{1, \dots, \lambda\}$ exclusively corrupts at most $t_j$ channels out of the $n$ channels for $t_j \geq 1$.
We have $\sum_{j=1}^\lambda t_j \leq n$.

\paragraph{The SMT Game.}
For an SMT protocol $\Pi$, we define our SMT game $\gamem(\Pi, \calA_1, \dots, \calA_\lambda)$ against $\lambda$ adversaries with the strategy profile
$(\calA_1, \dots, \calA_\lambda)$.
First, set parameters $\suc = 0$ and $\guess_j = \detect_j = 0$ for every $j \in \{1,\dots,\lambda\}$.
For the message space $\calM$ of $\Pi$, choose $m \in \calM$ uniformly at random,
and run the protocol $\Pi$ in which the message to be sent is $M_S = m$.
In the protocol execution, the sender or the receiver may send a special message ``DETECT at $i$'' for $i \in \{1, \dots, n\}$,
meaning that some tampering was detected in channel $i$.
Then, if adversary $j \in \{1, \dots, \lambda\}$ corrupts channel $i$, set $\detect_j = 1$.
After running the protocol, the receiver outputs $M_R$, and each adversary $j$ outputs $M_j$ for $j \in \{1,\dots,\lambda\}$.
If $M_R = M_S$, set $\suc = 1$.
For $j \in \{1,\dots,\lambda\}$, if $M_j = M_S$, set $\guess_j = 1$.
The outcome of the game is $\left(\suc, \{\guess_{j'},\detect_{j'}\}_{j'\in\{1,\dots,\lambda\}}\right)$.

\begin{definition}[Utility]
  The utility $U_j(\calA_1, \dots, \calA_\lambda,U)$ of adversary $j$ when the strategy profile $(\calA_1,\dots,\calA_\lambda)$
  and utility function $U$ are employed is the expected value $\E[U(j, \out)]$, where $U$ is a function that maps index $j$ and the outcome
  $\out = \left(\suc, \{\guess_{j'},\detect_{j'}\}_{j'\in\{1,\dots,\lambda\}}\right)$ of the game $\gamem(\Pi, \calA_1, \dots, \calA_\lambda)$
  to real values and the probability is taken over the randomness of the game.
\end{definition}

We define the security of SMT protocols against multiple adversaries. 
For strategies $\calB_1, \dots, \calB_\lambda$, and $\calA_j$,
we denote by $(\calA_j,\calB_{-j})$ the strategy profile $(\calB_1, \dots, \calB_{j-1}, \calA_j, \calB_{j+1}, \dots, \calB_\lambda)$.

\begin{definition}[PSMT against Multiple Rational Adversaries]\label{def:rsmtsec_mult}
  An SMT protocol $\Pi$ is \emph{perfectly secure against rational $(t_1, \dots, t_\lambda)$-adversaries with utility function $U$} if
  there are $t_j$-adversary $\calB_j$  for $j \in \{1, \dots, \lambda\}$ such that
    \begin{enumerate}
  \item \emph{Perfect security}: $\Pi$ is $(0,0)$-SMT against $(\calB_1, \dots, \calB_\lambda)$, and
  \item \emph{Nash equilibrium}: $U_j(\calA_j,\calB_{-j},U) \leq U_j(\calB_j,\calB_{-j},U)$
    for any $t_j$-adversary $\calA_j$ for every $j \in \{1, \dots, \lambda\}$.
  \end{enumerate}
\end{definition}

As in the analysis of protocols against a single adversary, we will consider a \emph{random guessing} strategy profile
$(\calB_1,\dots,\calB_\lambda)$ in which each $\calB_j$ is a random guessing adversary.
The perfect security for such a strategy profile immediately follows if the protocol satisfies the correctness.

\subsection*{Timid Adversaries}

For the case of multiple adversaries, we define timid adversaries who prefer other adversaries' tampering to be detected.
This property makes rational adversaries avoid cooperating with each other.
Let $U_\mathsf{timid}^\mathsf{mult}$ be the set of utility functions that satisfy the following three conditions:
\begin{enumerate}
\item $U(j,\out) > U(j,\out')$ if $\suc < \suc'$, $\guess_j = \guess_j'$,
  and $\detect_{j} = \detect_{j}'$; 
\item $U(j,\out) > U(j,\out')$ if $\suc = \suc'$, $\guess_j = \guess_j'$, 
  $\detect_j < \detect_j'$, and $\detect_{k} = \detect_{k}'$  for every $k \in \{1, \dots, \lambda\} \setminus \{j\}$; and
\item $U(j,\out) > U(j,\out')$ if $\suc = \suc'$, $\guess_j = \guess_j'$,  
  $\detect_{k} > \detect_{k}'$ for some $k \neq j$, and
  $\detect_{j'} = \detect_{j'}'$ for every $j' \in \{1, \dots, \lambda\} \setminus \{k\}$, 
\end{enumerate}
where $\out = \left(\suc, \{\guess_j,\detect_j\}_{j\in\{1,\dots,\lambda\}}\right)$
and $\out' = \left(\suc', \{\guess_j',\detect_j'\}_{j\in\{1,\dots,\lambda\}}\right)$
are the outcomes of the SMT game.

Let $U_\mathsf{st \mhyphen timid}^\mathsf{mult}$ be the set of utility functions satisfying the following condition in addition
to the above three:
\begin{enumerate}
   \setcounter{enumi}{3}
\item $U(j,\out) > U(j,\out')$ if $\suc > \suc'$, $\guess_j = \guess_j'$, $\detect_j < \detect_j'$, and $\detect_{k} = \detect_{k}'$
  for every $k \in \{1, \dots, \lambda\} \setminus \{j\}$.
\end{enumerate}

An adversary is said to be \emph{timid} if his utility function is in $U_\mathsf{timid}^\mathsf{mult}$,
and \emph{strictly timid} if the utility function is in $U_\mathsf{st\mhyphen timid}^\mathsf{mult}$.
For $j \in \{1,\dots,n\}$ and $b \in \bin$,
we write $\detect_{-j}=b$ if $\detect_{j'} = b$ for every $j' \in \{1,\dots, n\} \setminus \{j\}$.
In the analysis of  our protocols, we use the following values of the utility of adversary $j \in \{1,\dots,\lambda\}$.
\begin{itemize}
\item $u_1'$ is the utility when $\Pr[\guess_j = 1] = \frac{1}{|\calM|}$,   $\suc = 0$, $\detect_j=0$, $\detect_{-j} = 0$;
\item $u_2'$ is the utility when $\Pr[\guess_j = 1] = \frac{1}{|\calM|}$,   $\suc = 1$, $\detect_j=0$, $\detect_{-j} = 0$;
\item $u_3''$ is the utility when $\Pr[\guess_j = 1] = \frac{1}{|\calM|}$,   $\suc = 0$, $\detect_j=1$, $\detect_{-j} = 1$;
\item $u_3'$ is the utility when $\Pr[\guess_j = 1] = \frac{1}{|\calM|}$,   $\suc = 0$, $\detect_j=1$, $\detect_{-j} = 0$;
\item $u_4'$ is the utility when $\Pr[\guess_j = 1] = \frac{1}{|\calM|}$,   $\suc = 1$, $\detect_j = 1$, $\detect_{-j} = 0$.
\end{itemize}
For any utility function in $U_\mathsf{timid}^\mathsf{mult}$,
it holds that $u_1' > \max\{u_2', u_3''\}$, $\min\{u_2',u_3'\} > u_4'$, and $u_3'' > u_3'$.
If the utility is in $U_\mathsf{st \mhyphen timid}^\mathsf{mult}$, it holds that $u'_1 > u'_2 > u_3'' > u'_3 > u'_4$.

\section{Protocols against Multiple Adversaries}\label{sec:multadv}

\subsection{Protocol by Public Discussion}\label{sec:sjst}

We show that the SJST protocol of~\cite{SJST11} gives a perfect SMT-PD protocol against multiple adversaries.
As in Section~\ref{sec:single_sjst},
we modify the SJST protocol such that in the second and the third rounds,
if $b_i = 1$ in $B$ or $v_i = 1$ in $V$ for some $i \in \{1, \dots, n\}$,
the special message ``DETECT at $i$'' is also sent together.

\begin{theorem}\label{thm:sjst}
  For any  $\lambda \geq 2$, let $t_1, \dots, t_\lambda$ be integers
  satisfying $t_1 + \dots + t_\lambda \leq n$ and $1\leq t_i \leq n-1$ for every $i \in \{1, \dots, \lambda\}$.
  If the parameter $\ell$ in the SJST protocol satisfies
  \begin{equation*}
    \ell \geq \max_{t \in \{t_1, \dots, t_\lambda\}}\left\{1 + \log_2 t + \log_2 \frac{u'_3 - u'_4}{u'_2 - u'_4 - \alpha}, 1 + \frac{1}{t} \log_2 \frac{u'_1-u'_3}{\alpha} \right\}
  \end{equation*}
  for some  $\alpha \in (0, u'_2 - u'_4)$,
  then the protocol is perfectly secure against rational $(t_1, \dots, t_\lambda)$-adversaries with utility function $U \in U_\mathsf{timid}^\mathsf{mult}$.
\end{theorem}
\begin{proof}
  We note that the same argument as the proof of Theorem~\ref{thm:sjst_single} can apply to the case of multiple adversaries.
  This is because even in the presence of multiple adversaries, as long as we consider a Nash equilibrium,
  each adversary $j$ tries to maximize the utility by choosing a strategy $\calA_j$
  by assuming that all the other adversaries follow the random guessing strategy profile $\calB_{-j}$.
  It is precisely the case analyzed in the proof of Theorem~\ref{thm:sjst_single}.
  Since the utility values $u'_1, u'_2, u'_3, u_4'$ corresponds to those of $u_1, u_2, u_3, u_4$, respectively, in Theorem~\ref{thm:sjst_single}, the statement follows.
\end{proof}

\subsection{Protocol for Minority Corruptions}\label{sec:minor}

We provide a non-interactive SMT protocol based on secret-sharing and pairwise independent hash functions.
See Section~\ref{sec:rss} and Appendix~\ref{sec:uh} for the definitions.
The protocol is secure against multiple adversaries who only corrupt minorities of the channels.
Namely, we assume that each adversary corrupts at most $\lfloor (n-1)/2 \rfloor$ channels.
Note that the protocol does not use the public channel as in the protocol in Section~\ref{sec:sjst}.

We describe the construction of our protocol.
The protocol can employ any secret-sharing scheme of threshold $\lfloor (n-1)/2 \rfloor$,
which may be Shamir's scheme. 
Let $(s_1, \dots, s_n)$ be the shares generated by the scheme from the message to be sent.
Then, pairwise independent hash functions $h_i$ are chosen for each $i \in \{1, \dots, n\}$.
For any $j \neq i$, $h_i(s_j)$ is computed as an authentication tag for $s_j$.
Then, $(s_i, h_i, \{h_i(s_j)\}_{j \neq i})$ will be sent through the $i$th channel.
When $s_i$ is modified to $s_i' \neq s_i$ by some adversary,
the modification can be detected by the property of
pairwise independent hash functions because the adversary cannot modify all tags $h_j(s_i)$ for $j \neq i$.
Also, a random mask $r_{i,j}$ is applied to $h_i(s_j)$ to conceal the information of $s_j$ in $h_i(s_j)$.
The masks $\{r_{j,i}\}_{j \neq i}$ for $s_i$ will be sent through the $i$th channel so that only the $i$th channel reveals the information of $s_i$.
Hence, the message sent through the $i$th channel is $(s_i,  h_i, \{ h_i(s_j) \oplus r_{i,j}\}_{j \neq i}, \{r_{j,i}\}_{j \neq i})$.
As long as each adversary corrupts minorities of the channels,
a single adversary cannot cause erroneous detection of silent adversaries.
We give a formal description in Figure~\ref{fig:ciss}.

\begin{figure}[t]
  \begin{framed}
Let $(\Share, \Reconst)$ be a secret-sharing scheme of threshold $\lfloor (n-1)/2 \rfloor$,
  where a secret is chosen from $\calM$, and the shares are defined over $\mathcal{V}$.
  Let $m \in \calM$ be the message to be sent by the sender, and
  $H = \{h \colon \mathcal{V} \rightarrow \{0, 1\}^\ell\}$ a class of pairwise independent hash functions. 
  \begin{enumerate}
  \item The sender does the following:
    Generate the shares $(s_1, \dots, s_n)$ by $\Share(m)$, and randomly choose $h_i \in H$ for each $i \in \{1,\dots,n\}$.
    Also, for every distinct $i,j \in \{1,\dots,n\}$, choose $r_{i,j} \in \bin^\ell$ uniformly at random,
    and then compute $T_{i,j} = h_i(s_j) \oplus r_{i,j}$.
    Then, for each $i \in \{1,\dots,n\}$,
    send $m_i = \left(s_i, h_i, \{T_{i,j}\}_{j \in \{1,\dots, n\} \setminus \{i\}}, \{r_{j,i}\}_{j \in \{1, \dots, n\} \setminus \{i\}}\right)$
    through the $i$th channel.
  \item After receiving 
    $\tilde{m}_i = \left(\tilde{s}_i, \tilde{h}_i, \{\tilde{T}_{i,j}\}_{j \in \{1,\dots, n\} \setminus \{i\}}, \{\tilde{r}_{j,i}\}_{j \in \{1, \dots, n\} \setminus \{i\}}\right)$
    on each channel $i \in \{1, \dots, n\}$, the receiver does the following:
    For every $i  \in \{1,\dots,n\}$, compute the list
    $L_i = \left\{ j \in \{1, \dots, n\} \colon \tilde{h}_i(\tilde{s}_j) \oplus \tilde{r}_{i,j} \neq \tilde{T}_{i,j}\right\}$.
    If a majority of the lists coincide with a list $L$, reconstruct the message $\tilde{m}$
    by $\Reconst(\{i, \tilde{s}_i\}_{i \in \{1,\dots,n\} \setminus L})$, send messages ``DETECT at $i$'' for every $i \in L$, and output $\tilde{m}$.
    Otherwise, output $\bot$.
  \end{enumerate}
  \end{framed}
  \caption{\prot\label{pro:ciss} for Minority Corruption} \label{fig:ciss}
\end{figure}

\begin{theorem}\label{thm:ciss}
  For any  $\lambda \geq 2$, let $t_1, \dots, t_\lambda$ be integers
  satisfying $t_1 + \dots + t_\lambda \leq n$ and $1\leq t_i \leq \lfloor (n-1)/2 \rfloor$ for every $i \in \{1, \dots, \lambda\}$.
  If the parameter $\ell$ in Protocol~\ref{pro:ciss} satisfies
  \begin{equation*}
    \ell \geq \log_2\frac{u'_1-u'_4}{u'_2-u'_4} + 2 \log_2(n+1) -1,
  \end{equation*}
  then the protocol is perfectly secure against rational $(t_1, \dots, t_\lambda)$-adversaries with utility function $U \in U_\mathsf{timid}^\mathsf{mult}$.
\end{theorem}
\begin{proof}
  For $k \in \{1, \dots, \lambda\}$, let $\calB_{k}$ be a random guessing $t_k$-adversary.
  First, note that, for any $i \in \{1, \dots, n\}$, the information of $s_i$ can be obtained only by $m_i$, the message sent over the $i$th channel. 
  This is because for any $j \neq i$, $h_j(s_i)$ is masked as $h_j(s_i) \oplus r_{i,j}$,
  and the random mask $r_{i,j}$ is included only in $m_i$. 
  Also, each $s_i$ is a share of the secret sharing of threshold $\lfloor (n-1)/2 \rfloor$.
  Since $\calB_k$ can obtain at most $\lfloor (n-1)/2 \rfloor$ shares,
  $\calB_k$ can learn nothing about the message sent from the sender.
    Thus, the perfect security is achieved for $(\calB_1, \dots, \calB_\lambda)$.
  
  Next, we show that $(\calB_1, \dots, \calB_\lambda)$ is a Nash equilibrium.
  For $k \in \{1, \dots, \lambda\}$, let $\calA_k$ be any $t_k$-adversary.
  Since $U_k(\calB_1, \dots, \calB_\lambda) = u'_2$, to increase the utility, $\calA_k$ needs to get either (a) $\suc = 0$, or
  (b) $\detect_k=0$ and $\detect_{k'}=1$ for some $k' \neq k$.
  
  For the case of (a), $\calA_k$ tries to change $s_i$ into $\tilde{s}_i \neq s_i$ for some $i \in \{1, \dots, n\}$.
  Since $\calA_k$ does not corrupt some $i' \in \{1, \dots, n\}$,
  the index $i$ corrupted by $\calA_k$ will be included in the list $L_{i'}$ unless $h_{i'}(\tilde{s}_i) \oplus \tilde{r}_{i',i} = T_{i',i}$.
  Note that $\tilde{s}_i$ and $\tilde{r}_{i',i}$ are included in $\tilde{m}_i$, and thus can be changed,
  but $h_{i'}$ and $T_{i',i}$ are in $\tilde{m}_{i'}$, and thus have been unchanged.
  It follows from the property of pairwise independent hash functions that
  this can happen with probability $2^{1- \ell}$ assuming $\tilde{s}_i \neq s_i$.
  Thus, $i$ will be included in $L_{i'}$ with probability at least $1 - 2^{1-\ell}$.
  Since there are at least $n - \lfloor (n-1)/2 \rfloor = \lceil (n+1)/2 \rceil$ such indices $i'$,
  the probability that a majority of the lists contains $i$ is at least $1 - \lceil (n+1)/2 \rceil \cdot 2^{1-\ell}$.
  Note that $\calA_k$ may corrupt $\lfloor (n-1)/2 \rfloor$ channels in total.
  The probability that all the corrupted indices coincide with a majority of  the list is at least
  $1 - \lfloor (n-1)/2 \rfloor \cdot \lceil (n+1)/2 \rceil \cdot 2^{1-\ell} \geq 1 - (n+1)^2 \cdot 2^{-(\ell+1)}$.
  In that case, the message can be reconstructed by other shares, and thus we have $\suc = 1$, $\detect_k = 1$, and $\detect_{k'} = 0$ for $k' \neq k$,
  resulting in the utility of $u'_4$. Since $\calA_k$ only corrupts a minority of the channels,
  it cannot cause $\detect_{k'}=1$ for $k'\neq k$.
  Thus, the maximum utility of $\calA_k$ is $u'_1$.
  Thus, the utility of adversary $k$ when tampering as $\tilde{s}_i \neq s_i$ is at most
  \[ U_k(\calA_k, \calB_{-k}) \leq (n+1)^2\cdot 2^{-(\ell+1)} \cdot u'_1 + \left( 1 - (n+1)^2\cdot 2^{-(\ell+1)}\right) \cdot u'_4,\]
  which is at most $u'_2$ by the assumption on $\ell$.

  For the case of (b), $\calA_k$ needs to generate the corrupted message $\tilde{m}_i$ for the $i$th channel so that
  for a majority of indices $j \in \{1,\dots,n\}$, $\tilde{h}_i(s_j) \oplus r_{i,j} \neq \tilde{T}_{i,j}$, where each $j$ is corrupted by $\calB_{k'}$ with $k' \neq k$,
  and thus $r_{i,j}$ and $s_j$ are not tampered with.
  Since $\calA_k$ only corrupts a minority of the channels, this cannot happen.
  
  Therefore, $(\calB_1, \dots, \calB_\lambda)$ is a Nash equilibrium.
\end{proof}

Note that the single-adversary setting of Section~\ref{sec:singlegame} can be seen as a special case of
the multiple-adversary setting.
Namely, it is equivalent to the setting in which there are two adversaries, $\calA_1$ and $\calA_2$, such that
  $\calA_1$ tries to violate the security requirements of SMT by corrupting at most $t \leq n-1$ channels,
  whereas $\calA_2$, who corrupt $n-t \geq 1$ channels, does nothing for the protocol.
  Since Protocol~\ref{pro:ciss} does not rely on the additional utility of $u_3''$ in the security analysis,
  it also gives an SMT protocol in the single-adversary setting.
  \begin{corollary}\label{cor:ciss}
      If the parameter $\ell$ in Protocol~\ref{pro:ciss} satisfies
    \begin{equation*}
      \ell \geq \log_2\frac{u_1-u_4}{u_2-u_4} + 2 \log_2(n+1) -1,
    \end{equation*}
    then the protocol is perfectly secure against a rational $t$-adversary with utility function $U \in U_\mathsf{timid}$ for any $t < n/2$.
  \end{corollary}
  
\subsection{Protocol for Majority Corruptions}

We present a protocol against adversaries who may corrupt a majority of the channels.
We assume that adversaries are \emph{strictly} timid in this setting.
The protocol is a minor modification of the protocol for minority corruption.
In Protocol~\ref{pro:ciss}, the lists $L_i$ of the corrupted channels are generated for each channel,
and the final list $L$ is determined by the majority voting.
Thus, if an adversary corrupts a majority of the channels, the result of the majority voting can be easily forged,
and hence the protocol does not work for majority corruption.

To cope with majority corruptions, we modify the protocol such that
(1) the threshold of the secret sharing is changed from $\lfloor (n-1)/2 \rfloor$ to
$n-1$,
(2) the list $L_i$ contains both $i$ and $j$ if the masked tag $h_j(s_i) \oplus r_{i,j}$ does not match $T_{i,j}$,
and (3) the final list $L$ of the corrupted channels is composed of the union of all the sets $L_i$, namely, $L = L_1 \cup \dots \cup L_n$.
The threshold of $n-1$ can be achieved by Shamir's scheme.
Intuitively, this protocol works for strictly timid adversaries
because if some adversary tampers with messages over the $i$th channel,
the tampering will be detected with high probability, and in that case, $i$ must be included in the final list $L$. 
Since strictly timid adversaries prefer his tampering 
not to be detected, they will keep silent.
We give a formal description of the protocol in Figure~\ref{fig:ss-timid}.

\begin{figure}[t]
  \begin{framed}
  Let $(\Share, \Reconst)$ be a secret-sharing scheme of threshold $n-1$,
  where a secret is chosen from $\calM$, and the shares are defined over $\mathcal{V}$.
  Let $m \in \calM$ be the message to be sent by the sender, and
  $H = \{h \colon \mathcal{V} \rightarrow \{0, 1\}^\ell\}$ a class of pairwise independent hash functions. 
  \begin{enumerate}
  \item The sender does the following:
    Generate the shares $(s_1, \dots, s_n)$ by $\Share(m)$, and randomly choose $h_i \in H$ for each $i \in \{1,\dots,n\}$.
    Also, for every distinct $i,j \in \{1,\dots,n\}$, choose $r_{i,j} \in \bin^\ell$ uniformly at random,
    and then compute $T_{i,j} = h_i(s_j) \oplus r_{i,j}$.
    Then, for each $i \in \{1,\dots,n\}$,
    send $m_i = \left(s_i, h_i, \{T_{i,j}\}_{j \in \{1,\dots, n\} \setminus \{i\}}, \{r_{j,i}\}_{j \in \{1, \dots, n\} \setminus \{i\}}\right)$
    through the $i$th channel.
  \item After receiving 
    $\tilde{m}_i = \left(\tilde{s}_i, \tilde{h}_i, \{\tilde{T}_{i,j}\}_{j \in \{1,\dots, n\} \setminus \{i\}}, \{\tilde{r}_{j,i}\}_{j \in \{1, \dots, n\} \setminus \{i\}}\right)$
    on each channel $i \in \{1, \dots, n\}$, the receiver does the following:
    For every $i  \in \{1,\dots,n\}$, compute the list
    $L_i = \{ i \} \cup \left\{ j \in \{1, \dots, n\} \colon \tilde{h}_i(\tilde{s}_j) \oplus \tilde{r}_{i,j} \neq \tilde{T}_{i,j}\right\}$.
    Then, set $L = L_1 \cup \dots \cup L_n$.
    If $L = \emptyset$, reconstruct the message $\tilde{m}$ by $\Reconst(\{i, \tilde{s}_i\}_{i \in \{1,\dots,n\}})$,
    and output $\tilde{m}$.
    Otherwise, send messages ``DETECT at $i$'' for every $i \in L$, and output $\bot$ as the failure symbol.
  \end{enumerate}
  \end{framed}
  \caption{\prot\label{pro:ss-timid} for Majority Corruption} \label{fig:ss-timid}
\end{figure}

\begin{theorem}\label{thm:ciss2}
  For any  $\lambda \geq 2$, let $t_1, \dots, t_\lambda$ be integers
  satisfying $t_1 + \dots + t_\lambda \leq n$ and $1 \leq t_i \leq n-1$ for every $i \in \{1, \dots, \lambda\}$.
  If the parameter $\ell$ in Protocol~\ref{pro:ss-timid} satisfies
  \begin{equation*}
    \ell \geq \log_2\frac{u'_1-u''_3}{u'_2-u''_3} -1,
  \end{equation*}
  then the protocol is perfectly secure against rational $(t_1, \dots, t_\lambda)$-adversaries with utility function $U \in U_\mathsf{st\mhyphen timid}^\mathsf{mult}$.
\end{theorem}
\begin{proof}
  For $k \in \{1, \dots, \lambda\}$, let $\calB_k$ be a random guessing $t_k$-adversary.
  By the same reason as in the proof of Theorem~\ref{thm:ciss}, the protocol is perfectly secure against $(\calB_1, \dots, \calB_\lambda)$.

  Next, we  show that $(\calB_1, \dots, \calB_\lambda)$ is a Nash equilibrium.
  Let $\calA_k$ be any $t_k$-adversary for $k \in \{1,\dots,\lambda\}$.
  As in the proof of Theorem~\ref{thm:ciss}, $\calA_k$ needs to yield either (a) $\suc = 0$, or
  (b) $\detect_k=0$ and $\detect_{k'}=1$ for some $k' \neq k$.
  For the case of (a), $\calA_k$ needs to corrupt the $i$th channel so that $\tilde{s}_i \neq s_i$. 
  There is at least one index $i' \in \{1,\dots,n\}$ that is not corrupted by $\calA_k$. 
  Thus, by the property of pairwise independent hash functions, the index $i$ is included in the list $L_{i'}$  with probability at least $1-2^{1-\ell}$,
  in which case the utility of $\calA_k$ is at most $u''_3$.
  Hence, the expected utility is at most 
  \[ U_k(\calA_k, \calB_{-k}) \leq 2^{-(\ell+1)} \cdot u'_1 + \left( 1 - 2^{-(\ell+1)}\right) \cdot u''_3,\]
  which is at most $u'_2$ by assumption.
  For the case of (b), if some index $i$ is in the final list $L$ by the tampering by $\calA_k$,
  then some channel $i' \neq i$ corrupted by $\calA_k$ is also included in $L$.
  Thus, (b) cannot happen.
  Therefore, $(\calB_1, \dots, \calB_\lambda)$ is a Nash equilibrium.
\end{proof}

\section{SMT against Malicious and Rational Adversaries}\label{sec:mixed}

In the previous sections, we have discussed SMT against rational adversaries.
We have assumed that all the adversaries behave rationally.
The assumption may be strong because all of them can be characterized by the utility function we defined.
This section discusses more realistic situations where adversaries may not behave rationally but maliciously.

\subsection{Security Model}

Without loss of generality, we assume that there are $\lambda \geq 2$ adversaries, and adversaries $1, \dots, \lambda-1$ are rational,
and adversary $\lambda$ behaves maliciously.
We use the same definitions of the SMT game and the utility function in Section~\ref{sec:multgame}.
We define \emph{robust} security against rational adversaries.
A similar definition appeared in the context of rational secret sharing~\cite{ADGH06}.
For strategies $\calB_1, \dots, \calB_{\lambda-1}, \calA_\lambda, \calA_j$ for $j \in \{1,\dots,\lambda-1\}$,
we denote by $(\calA_j,\calB_{-j},\calA_\lambda)$ the strategy profile $(\calB_1, \dots, \calB_{j-1}, \calA_j, \calB_{j+1}, \dots, \calB_{\lambda-1}, \calA_\lambda)$.
\begin{definition}[Robust PSMT against Rational Adversaries]\label{def:r-rsmtsec}
  An SMT protocol $\Pi$ is \emph{$t^*$-robust perfectly secure against rational $(t_1, \dots, t_{\lambda-1})$-adversaries with utility function $U$} if
  there are $t_j$-adversary $\calB_j$  for $j \in \{1, \dots, \lambda-1\}$ such that for any $t_j$-adversary $\calA_j$ for $j \in \{1, \dots, \lambda-1\}$
  and $t^*$-adversary $\calA_\lambda$,
  \begin{enumerate}
  \item \emph{Perfect security}: $\Pi$ is $(0,0)$-SMT against $(\calB_1, \dots, \calB_{\lambda-1},\calA_\lambda)$, and
  \item \emph{Robust Nash equilibrium}: $U_j(\calA_j,\calB_{-j},\calA_\lambda,U) \leq U_j(\calB_j,\calB_{-j},\calA_\lambda,U)$ for every $j \in \{1,\dots,\lambda-1\}$ in the SMT game.
  \end{enumerate}
\end{definition}

Compared to Definition~\ref{def:rsmtsec_mult}, robust PSMT requires that the perfect security is achieved even in the presence of a malicious adversary $\calA_\lambda$,
and a strategy profile $(\calB_1, \dots, \calB_{\lambda-1},\calA_\lambda)$ is a Nash equilibrium for adversary $j \in \{1, \dots, \lambda-1\}$.

\subsection{Protocol against Malicious and Rational Adversaries}

We show that a robust PSMT protocol can be constructed based on the protocol for minority corruption in Section~\ref{sec:minor}.
For $t^*$-robust against $(t_1, \dots, t_{\lambda-1})$-adversaries, we assume that $t^* \leq \lfloor (n-1)/3 \rfloor$ and
$1 \leq t_j \leq \min\{ \lfloor (n-1)/2 \rfloor - t^*, (n-1)/3 \rfloor\}$
for each $j \in \{1, \dots, \lambda-1\}$.
Our non-interactive protocol is obtained simply by modifying the threshold of the secret sharing in Protocol~\ref{pro:ciss} 
from $\lfloor (n-1)/2 \rfloor$ to $\lfloor (n-1)/3 \rfloor$.
This protocol works because when only a malicious adversary corrupts at most $\lfloor (n-1)/3 \rfloor$ channels,
the transmission failure does not occur due to the error-correction property of the secret sharing.
Thus, perfect security is achieved in the presence of a malicious adversary.
Even if some rational adversary deviates from the protocol together with a malicious adversary,
they can affect at most $t_j + t^* \leq \lfloor (n-1)/2 \rfloor$ votes.
Thus, the majority voting can identify any tampering with high probability.

The formal description is given in Figure~\ref{fig:ss-robust}.
\begin{figure}[t]
  \begin{framed}
  Let $(\Share, \Reconst)$ be a secret-sharing scheme of threshold $\lfloor (n-1)/3 \rfloor$,
  where a secret is chosen from $\calM$, the shares are defined over $\mathcal{V}$,
  and the secret can be reconstructed as long as at most $\lfloor (n-1)/3 \rfloor$ out of $n$ shares are tampered.
  Let $m \in \calM$ be the message to be sent by the sender, and
  $H = \{h \colon \mathcal{V} \rightarrow \{0, 1\}^\ell\}$ a class of pairwise independent hash functions. 
  \begin{enumerate}
  \item The sender does the following:
    Generate the shares $(s_1, \dots, s_n)$ by $\Share(m)$, and randomly choose $h_i \in H$ for each $i \in \{1,\dots,n\}$.
    For every distinct $i,j \in \{1,\dots,n\}$, choose $r_{i,j} \in \bin^\ell$ uniformly at random,
    and then compute $T_{i,j} = h_i(s_j) \oplus r_{i,j}$.
    For each $i \in \{1,\dots,n\}$,
    send $m_i = \left(s_i, h_i, \{T_{i,j}\}_{j \in \{1,\dots, n\} \setminus \{i\}}, \{r_{j,i}\}_{j \in \{1, \dots, n\} \setminus \{i\}}\right)$ through the $i$th channel.
  \item After receiving 
    $\tilde{m}_i = \left(\tilde{s}_i, \tilde{h}_i, \{\tilde{T}_{i,j}\}_{j \in \{1,\dots, n\} \setminus \{i\}}, \{\tilde{r}_{j,i}\}_{j \in \{1, \dots, n\} \setminus \{i\}}\right)$
    on each channel $i \in \{1, \dots, n\}$, the receiver does the following:
    For every $i \in \{1, \dots, n\}$, compute the list 
    $L_i = \left\{ j \in \{1, \dots, n\} \colon \tilde{h}_i(\tilde{s}_j) \oplus \tilde{r}_{i,j} \neq \tilde{T}_{i,j}\right\}$. 
    If a majority of the lists coincide with a list $L$, 
    reconstruct the message $\tilde{m}$ by $\Reconst(\{i, \tilde{s}_i\}_{i \in \{1,\dots,n\}})$, send message ``DETECT at $i$'' for every $i \in L$,
    and output $\tilde{m}$.
    Otherwise, output $\bot$.
  \end{enumerate}
\end{framed}
  \caption{\prot\label{pro:ss-robust}} \label{fig:ss-robust}
\end{figure}

For the security analysis, we define the values of the utility of adversary $j \in \{1,\dots,\lambda-1\}$ such that
\begin{itemize}
\item $u_1''$ is the utility in the same case as $u_1'$ except that $\detect_\lambda = 1$,
\item $u_2''$ is the utility in the same case as $u_2'$ except that $\detect_\lambda = 1$, and
\item $u_4''$ is the utility in the same case as $u_4'$ except that $\detect_\lambda = 1$.
\end{itemize}
The values $u_1', u_2', u_4'$ are defined as the case that $\detect_{j'}=0$ for every $j' \in \{1,\dots,\lambda\} \setminus \{j\}$.
In the above, the values $u_1'', u_2'', u_4''$ are defined as $\detect_{j'}=0$ for every $j' \in \{1,\dots,\lambda-1\} \setminus \{j\}$ and $\detect_\lambda = 1$.

\begin{theorem}\label{thm:ss-robust}
  For any  $\lambda \geq 2$, let $t_1, \dots, t_{\lambda-1}, t^*$ be integers
  satisfying $t_1 + \dots + t_{\lambda-1}+t^* \leq n$, $0 \leq t^* \leq \lfloor (n-1)/3 \rfloor$,
  and $1\leq t_i \leq \min\{ \lfloor (n-1)/2 \rfloor - t^*, \lfloor (n-1)/3 \rfloor\}$ for every $i \in \{1, \dots, \lambda-1\}$.
  If the parameter $\ell$ in Protocol~\ref{pro:ss-robust} satisfies
  \begin{equation*}
    \ell \geq \max_{(u_1^*,u_2^*,u_4^*) \in \{(u_1',u_2',u_4'),(u_1'',u_2'',u_4'')\}}\left\{ \log_2\frac{u_1^*-u_4^*}{u_2^*-u_4^*} + 2 \log_2(n+1) -1 \right\},
  \end{equation*}
  then the protocol is $t^*$-robust perfectly secure against rational $(t_1, \dots, t_{\lambda-1})$-adversaries with utility function $U \in U_\mathsf{timid}^\mathsf{mult}$.
\end{theorem}
\begin{proof}
  For $k \in \{1, \dots, \lambda-1\}$, let $\calB_{k}$ be a random guessing adversary.
Let $\calA_\lambda$ be any $t^*$-adversary.
Note that the information of $s_i$ can be obtained only by seeing $m_i$ since each $h_j(s_i)$ is masked by $r_{j,i}$, which is included only in $m_i$.
Since each $s_i$ is a share of the secret sharing of threshold $\lfloor (n-1)/3 \rfloor$,
each adversary $\calB_k$ and $\calA_\lambda$ can learn nothing about the original message.
Although at most $t^*$ messages may be corrupted by $\calA_\lambda$, it follows from the property of the underlying secret sharing
that the message can be correctly recovered in the presence of $t^* \leq \lfloor (n-1)/3 \rfloor$ corruptions out of $n$ shares.
Thus, the protocol is perfectly secure against $(\calB_1, \dots, \calB_{\lambda-1},\calA_\lambda)$.

Next, we show that $(\calB_1, \dots, \calB_{\lambda-1},\calA_\lambda)$ is a Nash equilibrium for any $\calA_\lambda$.
When the strategy profile $(\calB_1, \dots, \calB_{\lambda-1},\calA_\lambda)$ is employed, we have $\suc = 1$.
To increase the utility of adversary $k$, 
$\calA_k$ needs to get either (a) $\suc = 0$, or
  (b) $\detect_k=0$, and $\detect_{k'}=1$ for some $k' \neq k$.

For the case of (a), $\calA_k$ tries to change $s_i$ into $\tilde{s}_i \neq s_i$ for some $i \in \{1, \dots, n\}$.
When playing with $(\calA_k, \calB_{-k},\calA_\lambda)$, the number of corrupted channels is at most $t_k + t^* \leq \lfloor (n-1)/2 \rfloor$.
Hence, there are a majority of indices $i'$ that is not corrupted by $\calA_k$ or $\calA_\lambda$,
and for each $i'$, the tampering on the $i$th channel will be detected; namely, the list $L_{i'}$ will include $i$ with high probability.
By the same argument as in the proof of Theorem~\ref{thm:ciss}, any tampering of $\tilde{s}_i \neq s_i$ by $\calA_k$ and $\calA_\lambda$
is detected with probability at least $1 - (n+1)^2\cdot 2^{-(\ell+1)}$.
Thus, we have that
\begin{align*}
  U_k(\calA_k, \calB_{-k},\calA_\lambda) & \leq (n+1)^2\cdot 2^{-(\ell+1)} \cdot u_1^* + \left( 1 - (n+1)^2\cdot 2^{-(\ell+1)}\right) \cdot u_4^* 
   \leq u_2^*,
  \end{align*}
where $(u_1^*,u_2^*, u_4^*)$  is either $(u_1',u_2',u_4')$ or $(u_1'',u_2'',u_4'')$.
The last inequality follows from the assumption.

For the case of (b), $\calA_k$ needs the result that $j \in L_i$ for a majority of the list $L_i$'s, where the $j$th channel is corrupted by adversary $k'$.
However, 
since $\calA_k$ and $\calA_\lambda$ can corrupt a minority of the channels,
this event cannot happen.

Thus, we have shown that $(\calB_1, \dots, \calB_{\lambda-1})$ is a robust Nash equilibrium. 
\end{proof}

\section{Conclusions}\label{sec:conclusion}

We have introduced game-theoretic security models in SMT and constructed perfect SMT protocols against rational timid adversaries.
Several protocols could circumvent the known impossibility results in the traditional cryptographic model.
We have also constructed perfect SMT protocols when multiple rational adversaries corrupt all the channels.
The results have revealed that we may not need to guarantee that
adversaries do not corrupt one resource/channel if they may not cooperate.
A feature of our model  is that the best strategy for adversaries is to behave harmlessly.
Namely, adversaries rationally decide to do nothing for the protocols.
Although this conclusion seems similar to the honest-but-curious adversary model,
the difference is significant between the situations in which adversaries can potentially attack actively or not.

One of future work is to apply our game-theoretic models to other primitives and protocols.
The model of timid adversaries might be useful for constructing more efficient and resilient protocols.
It can be used to avoid the impossibility results in the traditional setting.
Since rational adversaries in our models do not attack actively,
it seems easier to construct protocols by composition.
Another direction is  to study the mixed model of malicious and rational adversaries.
Since real-life situations may fall into this setting, it is beneficial to construct more efficient
protocols than in the usual cryptographic setting.

\section*{Acknowledgments}
This work was supported in part by JSPS Grants-in-Aid for Scientific Research Numbers 16H01705, 17H01695, 18K11159, 19K22849, and 21H04879.
We thank Masahito Hayashi for discussions about the relations between SMT and secure network coding.
  
\bibliographystyle{abbrv}
\bibliography{mybib}

%\clearpage

%\appendices
\appendix

%\section{Building Blocks}

\section{Pairwise Independent Hash Functions}\label{sec:uh}

Wegman and Carter \cite{WC81} introduced the notion of pairwise independent (or strongly universal)
hash functions and gave its construction.

\begin{definition}
Suppose that  a class of hash functions 
$H = \{h \colon \{0,1\}^m \rightarrow \{0, 1\}^\ell\}$, where 
$m \ge \ell$, satisfies the following:
for any distinct $x_1,x_2\in \{0,1\}^m$ and $y_1,y_2\in\{0,1\}^\ell$,
\[ \Pr_{h\in H} [ h(x_1)=y_1 \land h(x_2)=y_2 ] \le \gamma. \]
Then $H$ is called \emph{$\gamma$-pairwise independent}.
In the above, the randomness comes from the uniform choice of $h$ over $H$.
\end{definition}

Here we mention a useful property of almost pairwise independent hash function, 
which guarantees the security of some SMT protocols.

\begin{lemma}[\cite{SJST11}]\label{lem:sjst}
Let $H = \{h \colon \{0,1\}^m \rightarrow \{0, 1\}^\ell\}$ be a $\gamma$-almost
pairwise independent hash function family. Then for any 
$(x_1,c_1)\ne (x_2,c_2)\in \{0,1\}^m\times \{0,1\}^\ell$, we have
\[ \Pr_{h\in H} [ c_1\oplus h(x_1) = c_2\oplus h(x_2) ] \le 2^{\ell}\gamma. \]
\end{lemma}

In \cite{WC81}, Wegman and Carter constructed a family of 
$2^{1-2\ell}$-almost pairwise independent
hash functions.
In particular, their hash function family 
$H_{wc} = \{ h \colon \{0,1\}^m\rightarrow \{0,1\}^{\ell} \}$
satisfies that
\[ \Pr_{h\in H_{wc}} [ h(x_1)=y_1 \land h(x_2)=y_2 ] = 2^{1-2\ell} \]
for any distinct $x_1,x_2\in\{0,1\}^m$ and for any $y_1,y_2\in \{0,1\}^{\ell}$
and also
\begin{equation}
  \Pr_{h\in H_{wc}} [ c_1\oplus h(x_1) = c_2 \oplus h(x_2) ] = 2^{1-\ell}\label{eq:wc}
\end{equation}
for any distinct pairs $(x_1,c_1)\ne (x_2,c_2)\in \{0,1\}^m\times \{0,1\}^{\ell}$.

\section{Proof of Theorem~\ref{thm:robustss}}\label{sec:proofrobustss}

To prove the theorem, we define the notion of \emph{algebraic manipulation detection (AMD) codes}
in which the security requirement is slightly different from that in~\cite{CDFPW08} for our purpose.

\begin{definition}\label{def:amd}
  An \emph{$(M, N, \delta)$-algebraic manipulation detection (AMD) code} is a probabilistic function $E \colon \mathcal{S} \to \mathcal{G}$,
  where $\mathcal{S}$ is a set of size $M$
  and $\mathcal{G}$ is an additive group of order $N$,
  together with a decoding function $D \colon \mathcal{G} \to \mathcal{S} \cup \{\bot\}$
  such that
  \begin{itemize}
  \item \emph{Correctness}: For any $s \in \mathcal{S}$, $\Pr[ D( E(s) ) = s] = 1$.
  \item \emph{Security}: For any $s \in \mathcal{S}$ and $\Delta \in \mathcal{G} \setminus \{0 \}$, $\Pr[D ( E(s) + \Delta) \neq \bot ] \leq \delta$.
  \end{itemize}
  
  An AMD code is called \emph{systematic} if $\mathcal{S}$ is a group, and the encoding is of the form
  \begin{equation*}
    E \colon \mathcal{S} \to \mathcal{S} \times \mathcal{G}_1 \times \mathcal{G}_2, s \mapsto (s, x, f(x,s))
  \end{equation*}
  for some function $f$ and random $x \in \mathcal{G}_1$.
  The decoding function $D$ of a systematic AMD code is given by $D( s', x', f') = s'$ if $f' = f(x',s')$, and $\bot$ otherwise.
\end{definition}
Note that, for a systematic AMD code, the correctness immediately follows from the definition of the decoding function.
The security requirement can be stated such that
for any $s \in \mathcal{S}$ and $(\Delta_s, \Delta_x, \Delta_f) \in \mathcal{S} \times \mathcal{G}_1 \times \mathcal{G}_2 \setminus \{ (0,0,0)\}$,
$\Pr_x[ f(s + \Delta_s, x + \Delta_x) = f(s, x) + \Delta_f] \leq \delta$.

We show that a systematic AMD code given in~\cite{CDFPW08} satisfies the above definition.
\begin{proposition}\label{prop:amd}
  Let $\F$ be a finite field of size $q$ and characteristic $p$, and $d$ any integer such that $d+2$ is not divisible by $p$.
  Define the encoding function $E \colon \F^d \to \F^d \times \F \times \F$ by $E(s) = (s, x, f(x,s))$ where 
  \[ f(x,s) = x^{d+2} + \sum_{i=1}^{d} s_i x^i\]
  and $s = (s_1, \dots, s_d)$.
  Then, the construction is a systematic $(q^d, q^{d+2}, (d+1)/q)$-AMD code.
\end{proposition}
\begin{proof}
  We show that for any $s \in \F^{d}$ and $(\Delta_s, \Delta_x, \Delta_f) \in \F^d \times \F \times \F \setminus \{(0^d,0,0)\}$,
  $\Pr[ f(s+\Delta_s, x+ \Delta_x) = f(s,x) + \Delta_f] \leq \delta$.
  The event in the probability is that
  \begin{equation}
    (x+\Delta_x)^{d+2} + \sum_{i=1}^d s_i'(x+\Delta_x)^i = x^{d+2} + \sum_{i=1}^d s_ix^i + \Delta_f,\label{eq:event}
  \end{equation}
  where $s_i'$ is the $i$th element of $s+\Delta_s$.
  The left-hand side of (\ref{eq:event}) can be represented by
  \begin{equation*}
    x^{d+2} + (d+2)\Delta_x x^{d+1} + \sum_{i=1}^d s_i' x^i + \Delta_x p(x)
  \end{equation*}
  for some polynomial $p(x)$ of degree at most $d$.
  Thus, (\ref{eq:event}) can be rewritten as
  \begin{equation}
    (d+2)\Delta_x x^{d+1} + \sum_{i=1}^d (s_i' - s_i) x^i + \Delta_x p(x) - \Delta_f = 0. \label{eq:event2}
  \end{equation}
  
  We discuss the probability that (\ref{eq:event2}) happens when $x$ is chosen uniformly at random.  We consider the following cases:
  \begin{enumerate}
  \item When $\Delta_x \neq 0$, the coefficient of $x^{d+1}$ is $(d+2) \Delta_x$, which is not zero by the assumption that $d+2$ is not divisible by $p$.
    Then, (\ref{eq:event2}) has at most $d+1$ solutions $x$. Hence the event happens with probability at most $(d+1)/q$.
  \item When $\Delta_x = 0$, we consider two subcases:
    \begin{enumerate}
    \item If $\Delta_s \neq 0$, then $s_i' - s_i \neq 0$ for some $i$. Hence (\ref{eq:event2}) has at most $d$ solutions $x$. Thus the event happens with probability at most $d/p$.
    \item If $\Delta_s = 0$, (\ref{eq:event2}) becomes $\Delta_f = 0$. Since $\Delta_f \neq 0$ for this case, the event cannot happen.
    \end{enumerate}
  \end{enumerate}
  In every case, the event happens with probability at most $(d+1)/q$. Thus the statement follows.
\end{proof}

As discussed in~\cite{CDFPW08}, a robust secret sharing scheme can be obtained by
combining an AMD code and a linear secret sharing scheme.
Let $(\Share, \Reconst)$ be a $(t,n)$-secret sharing scheme with range $\mathcal{G}$ that
satisfies correctness and perfect privacy of Definition~\ref{def:robustss},
where we drop the parameter $\delta$ for robustness.
A \emph{linear} secret sharing scheme has the property that for any $s \in \mathcal{G}$, $(s_1, \dots, s_n) \in \Share(s)$,
and vector $(s_1', \dots, s_n')$, which may contain $\bot$ symbols,
it holds that $\Reconst( \{i, s_i + s_i'\}_{i \in I}) = s + \Reconst( \{i, s_i'\}_{i \in I})$
for any $I \subseteq \{1, \dots, n\}$  with $|I| > t$,
where $\bot + x = x + \bot = \bot$ for all $x$.
Examples of linear secret sharing schemes are Shamir's scheme~\cite{Sha79} and the simple XOR-based $(n-1,n)$-scheme,
in which secret $s \in \bin^n$ is shared by $(s_1, \dots, s_n)$ for random $s_i\in\bin^n$ with the restriction that $s_1 \oplus \dots \oplus s_n = s$.

We show that the same construction as in~\cite{CDFPW08} works as a construction of robust secret sharing as per Definition~\ref{def:robustss}.
\begin{proposition}\label{prop:robustss}
  Let $(\Share, \Reconst)$ be a linear $(t,n)$-secret sharing scheme with range $\mathcal{G}$ that
  satisfies correctness and perfect privacy as per Definition~\ref{def:robustss},
  and let $(E,D)$ be an $(M,N,\delta)$-AMD code as per Definition~\ref{def:amd} with $|\mathcal{G}| = N$.
  Then, the scheme $(\Share', \Reconst')$ defined by $\Share'(s) = \Share(E(s))$ and $\Reconst'(S) = D(\Reconst(S))$ is
  a $(t,n,\delta)$-robust secret sharing scheme. 
\end{proposition}
\begin{proof}
  Let $(s_1, \dots, s_n) \in \Share'(s)$. Let $I \subseteq \{1,\dots,n\}$ with $|I| \leq t$,
  and $(\tilde{s}_1, \dots \tilde{s}_n)$ be a sequence of shares
  satisfying the requirement for input shares in the robustness condition of Definition~\ref{def:robustss}.
  We assume that $\tilde{s}_i = s_i + \Delta_i'$ for each $i \in \{1,\dots,n\}$. Note that $\Delta_i' = 0$ for every $i \notin I$.
  Then,
  \begin{align*}
    \Pr\left[ \Reconst'\left(\{i,\tilde{s}_i\}_{i \in \{1,\dots,n\}}\right) \neq \bot \right]
    & = \Pr\left[ D\left( E(s) + \Reconst(\{ i, \Delta_i\}_{i \in \{1,\dots,n\}}) \right) \neq \bot \right]\\
    & = \Pr\left[D\left( E(s) + \Delta \right) \neq \bot \right],
  \end{align*} 
  where $\Delta = \Reconst\left(\{ i, \Delta_i\}_{i \in \{1,\dots,n\}}\right)$ is determined by the adversary.
  It follows from perfect privacy of the secret sharing scheme that $\Delta$ is independent of $E(s)$.
  Thus, if $\tilde{s}_i \neq s_i$ for some $i \in \{1, \dots, n\}$, 
  the probability is at most $\delta$ by the security of the AMD code. Hence, the statement follows.
\end{proof}

By combining Shamir's secret sharing scheme with range $\F^d$ and the AMD code of Proposition~\ref{prop:amd},
the robust secret sharing scheme of Theorem~\ref{thm:robustss} is obtained by Proposition~\ref{prop:robustss}.

\end{document}